\documentclass[11pt]{amsart}
\usepackage{amssymb,mathrsfs,graphicx,enumerate}
\usepackage{amsmath,amsfonts,amssymb,amscd,amsthm,bbm}
\usepackage[retainorgcmds]{IEEEtrantools}
\usepackage{tikz-cd}
\usepackage{makecell}
\usepackage{tikz}
\usepackage{array}
\usetikzlibrary{matrix}
\usepackage{pgfpages}
\usepackage{colortbl}

\usepackage{tabto}
\usepackage{pgfplots}
\usepackage{asypictureB}
\usepackage{graphicx, subfigure}

\title[The Lohe Hermitian sphere model with frustration]{Emergent dynamics of the Lohe Hermitian sphere model with frustration}

\author[Ha]{Seung-Yeal Ha}
\address[Seung-Yeal Ha]{\newline Department of Mathematical Sciences\newline Seoul National University, Seoul 08826 and \newline
Korea Institute for Advanced Study, Hoegiro 85, 02455, Seoul, Republic of Korea}
\email{syha@snu.ac.kr}

\author[Kang]{Myeongju Kang}
\address[Myeongju Kang]{\newline Department of Mathematical Sciences\newline Seoul National University, Seoul 08826, Republic of Korea}
\email{bear0117@snu.ac.kr}

\author[Park]{Hansol Park}
\address[Hansol Park]{\newline Department of Mathematical Sciences\newline Seoul National University, Seoul 08826, Republic of Korea}
\email{hansol960612@snu.ac.kr}

\newtheorem{theorem}{Theorem}[section]
\newtheorem{lemma}{Lemma}[section]

\newtheorem{proposition}{Proposition}[section]
\newtheorem{remark}{Remark}[section]

\newtheorem{definition}{Definition}[section]

\newcommand{\bbc}{\mathbb C}
\newcommand{\bbh}{\mathbb H}

\newcommand{\bbr}{\mathbb R}
\newcommand{\bbs}{\mathbb S}

\begin{document}

\date{\today}


{\thanks{\textbf{Acknowledgment.} The work of S.-Y. Ha was supported by National Research Foundation of Korea (NRF-2020R1A2C3A01003881), the work of M. Kang was supported by the National Research Foundation of Korea(NRF) grant funded by the Korea government(MSIP)(2016K2A9A2A13003815), and the work of H. Park was supported by Basic Science Research Program through the National Research Foundation of Korea(NRF) funded by the Ministry of Education (2019R1I1A1A01059585)}

\begin{abstract}
We study emergent dynamics of the Lohe hermitian sphere(LHS) model which can be derived from the Lohe tensor model \cite{H-P2} as a complex counterpart of the Lohe sphere(LS) model. The Lohe hermitian sphere model describes aggregate dynamics of point particles on the hermitian sphere $\bbh\bbs^d$ lying in ${\mathbb C}^{d+1}$, and the coupling terms in the LHS model consist of two coupling terms. For identical ensemble with the same free flow dynamics, we provide a sufficient framework leading to the complete aggregation in which all point particles form a giant one-point cluster asymptotically. In contrast, for non-identical ensemble, we also provide a sufficient framework for the practical aggregation. Our sufficient framework is formulated in terms of coupling strengths and initial data. We also provide several numerical examples and compare them with our analytical results. 
\end{abstract}
\maketitle 

\centerline{\date}


\section{Introduction}  \label{sec:1}
\setcounter{equation}{0} 
Emergent behaviors of complex systems are ubiquitous in nature, for example, flocking of birds, swarming of fish, flashing of fireflies and herding of sheep, etc \cite{A-B, A-B-F, B-D-P, B-H, B-T1, B-T2, B-B, H-K-P-Z, Pe, P-R, St, T-B-L, T-B, VZ, Wi1}. Several jargons such as aggregation, flocking, synchronization and herding are often used to describe such collective behaviors. Before we go into our topics, we briefly review several basic terminologies and concepts to be used in this paper. $\bbc$ denotes the complex field and let $\bbc^d$ be the cartesian product of $d$ copies of $\bbc$ for a positive integer $d$. Thus, the points of $\bbc^d$ are ordered  $d$-tuples $z = (z^1, \cdots, z^d)$ where $z_i \in \bbc$.  Algebraically, $\bbc^d$ is a $d$-dimensional vector space over $\bbc$, and topologically it is the euclidean space $\bbr^{2d}$ of real dimension $2d$, so we may call it complex euclidean space. For $z = (z^1, \cdots, z^d)$ and $w = (w^1, \cdots, w^d)$ in $\bbc^d$, we define the inner product $\langle \cdot, \cdot \rangle$ and the associated norm $\| \cdot \|$:
\[ \langle z, w \rangle := \sum_{i=1}^{d} \bar{z}^i w^i, \quad \|z\| := \langle z, z \rangle^{\frac{1}{2}}, \]
where we used physicist's notation by conjugating the first argument in $\langle \cdot, \cdot \rangle$. Let $v \in {\bbc}^{d+1}$ and $W \in {\bbc}^{(d+1) \times (d+1)}$ be a complex vector and  complex matrix, respectively. Then, we denote $i^{th}$-component and $(i,j)$-component of the real vector $v$ and real matrix $A$ by $[v]_i$ and $[A]_{ij}$, respectively. Moreover, $W^{\dagger} \in {\bbc}^{d \times d}$ and $\|W \|_F$ are the hermitian conjugate and norm of $W$:
\[ [W^{\dagger}]_{ij} = \overline{[W]_{ji}}, \quad 1 \leq i, j \leq d+1, \quad \| W \|_F := \mbox{Tr}(W^{\dagger} W)^\frac{1}{2}. \]
In this paper, we are interested in an aggregate phenomenon of a particle ensemble on the unit (hermitian) sphere $\bbh\bbs^d$ in ${\mathbb C}^{d +1}$ under the effect of frustration: 
\[ {\mathbb H}{\mathbb S}^d := \{ z \in {\mathbb C}^{d+1}:~ \|z \| = 1 \}. \] 
Here we used the adjective ``{\it hermitian}" to distinguish the unit sphere in $\bbc^{d+1}$ and the unit sphere in $\bbr^{d+1}$. 

For phase-coupled limit-cycle oscillator models such as the Kuramoto model and the Winfree model, (interaction) frustration often appear as a form of phase shift, and it generates diverse asymptotic patterns through the competitions between synchronizing enforcing terms and periodic enforcing terms. This is why the study of frustrated systems is so interesting from the viewpoint of nonlinear dynamics. For details, we refer to \cite{H-K-L, H-K-L1}. Recently, aggregation modelings for a particle ensemble on the unit sphere $\bbs^d$ in $\bbr^{d+1}$ has been extensively studied in literature \cite{C-C-H, C-H1, C-H2, C-H3, C-H5, J-C, Lo-1, Lo-2, M-T-G, T-M, Ol, Zhu} in the absence of frustration. 

To put our discussion in a proper setting, we begin with ``{\it the Lohe sphere(LS) model with frustration}" introduced in  \cite{H-K-P-R} . Let $x_j= x_j(t) \in \bbs^d$ be the position of the $j$-th Lohe particle. Then, the LS model in the presence of frustration reads as follows.
\begin{equation} \label{A-0}
\dot x_j = \Omega_j x_j + \frac{\kappa}{N} \sum_{k=1}^N \Big(Vx_k - \langle x_j, Vx_k\rangle x_j \Big),  \quad j = 1, \cdots, N,
\end{equation}
where $\Omega_j \in \bbr^{(d+1) \times (d+1)}$ is the natural frequency matrix of the $j$-th particle which is skew-symmetric $(\Omega^{t} = -\Omega)$, and $V \in \bbr^{(d+1) \times (d+1)}$ is the frustration matrix consisting of the sum of the identity matrix and skew-symmetric matrix $W$:
\begin{equation} \label{A-1}
 V = I_{d+1} + W \quad \mbox{and} \quad W^{t} = -W.
 \end{equation}
 Then, system \eqref{A-0} - \eqref{A-1} can be rewritten as follows:
 \begin{equation} \label{A-2}
\dot x_j = \Omega_j x_j + \frac{\kappa}{N} \sum_{k=1}^N \Big(x_k - \langle x_j, x_k\rangle x_j \Big) +  \frac{\kappa}{N} \sum_{k=1}^N \Big(Wx_k - \langle x_j, Wx_k\rangle x_j \Big).
\end{equation}
For $W = 0$, system \eqref{A-2} reduces to the Lohe sphere model on the complete graph, and its emergent dynamics has been studied in \cite{C-C-H, C-H1, C-H2, C-H3, C-H5, M-T-G, Ol, Zhu}. \newline

In this paper, we are interested in the following two questions:
\begin{itemize}
\item
(Q1):~What is the complex analogue of \eqref{A-0}?

\vspace{0.2cm}

\item
(Q2):~Can we rigorously verify emergent dynamics of the proposed complex counterpart?
\end{itemize}
\vspace{0.2cm}
In the absence of frustration, i.e., $V \equiv I_{d+1}$, the complex analogue of the Lohe sphere model has been proposed in \cite{H-K-P-R} and its emergent dynamics was also studied for identical ensemble. In what follows, we briefly summarize our main results on (Q1) - (Q2). \newline

First, we present the complex analog of system \eqref{A-0}.  Let $z_j = z_j(t)$ be the position of the $j$-th particle on the unit sphere in $\mathbb{HS}^d$. Then, the proposed Lohe hermitian sphere model with frustration reads as follows:
\begin{equation} \label{LHSF}
\dot{z}_j = \Omega_jz_j+\frac{\kappa_0}{N}\sum_{k=1}^N \Big(\langle z_j, z_j\rangle V_0z_k-\langle V_0z_k, z_j\rangle z_j \Big) +\frac{\kappa_1}{N}\sum_{k=1}^N \Big (\langle z_j, V_1z_k\rangle-\langle V_1z_k, z_j\rangle \Big)z_j.
\end{equation}
Here $\kappa_0$ and $\kappa_1$ are nonnegative constants, and the frustration matrices $V_0$ and $V_1$ take the same form as in \eqref{A-1}:
\begin{equation} \label{fr}
V_0 =  I_{d+1}+W_0, \quad V_1= I_{d+1}+W_1, 
\end{equation}
where $\Omega_j, W_0$ and $W_1$ are skew-hermitian matrices (see Section \ref{sec:2.1} for details):
\[ \Omega_j^{\dagger} = -\Omega_j, \quad j = 1, \cdots, N, \quad  W_0^{\dagger} = -W_0, \quad W_1^{\dagger} = -W_1. \]
Note that for real vector case, the term $\langle z_j, V_1z_k\rangle-\langle V_1z_k, z_j\rangle$ vanishes, and system \eqref{LHSF} reduces to system \eqref{A-0}. Hence, our proposed system \eqref{LHSF} can be called a complex counterpart of \eqref{A-0}. Moreover, it can be rewritten as a mean-field form:
\begin{equation} \label{A-2-1}
\dot{z}_j =\Omega_jz_j+ \kappa_0 \Big( \langle z_j, z_j\rangle V_0 z_c - \langle z_c, z_j \rangle z_j \Big) + \kappa_1 \Big( \langle z_j, V_1 z_c \rangle - \langle V_1 z_c, z_j \rangle         \Big) z_j,
\end{equation}
where $z_c = \frac{1}{N} \sum_{i=1}^{N} z_i$.

Second, we return to system \eqref{A-0}, and study emergent dynamics of the Lohe sphere model for non-identical ensemble in mean-field form:
\begin{equation} \label{A-3}
\dot{x}_j=\Omega_j x_j+\kappa_0\Big(\langle{x_j, x_j}\rangle V_0 x_c-\langle{V_0 x_c, x_j}\rangle x_j \Big), \quad x_c := \frac{1}{N} \sum_{j=1}^{N} x_j. 
\end{equation}
For the identical ensemble with $\Omega_j = \Omega$, emergent dynamics for \eqref{A-3} has been already studied in \cite{H-K-P-R} in which exponential aggregation was achieved using a position diameter as a suitable Lyapunov functional. Hence, the previous approach in aforementioned literature is global. In Section \ref{sec:3}, we revisit complete aggregation problem using a local approach. For this, we introduce an inter-particle angle $\theta_{ij}$ as follows.
\begin{equation*} \label{A-4}
\theta_{ij} :=\cos^{-1}\left(\langle x_i, x_j\rangle\right), \quad 1 \leq i, j \leq N.
\end{equation*}
For an identical ensemble, if the initial data $\Theta^{in}$ satisfies
\[
\theta_{ij}^{in}<\cot^{-1}\left(\frac{\|W\|_F}{\sqrt{2}}\right), \quad 1 \leq i, j \leq N.
\]
then there exists a positive constant $\Lambda_{ij} = \Lambda_{ij}(N, \kappa, W, \Theta^{in})$ such that 
\[
\theta_{ij}(t)\leq \theta_{ij}^{in}\exp\left(-\Lambda_{ij}t\right), \quad t \geq 0,
\]
which improves the earlier result in \cite{H-K-P-R} (see Proposition \ref{P3.1}).  In contrast, for non-identical ensemble, if the initial data $\{x_i^0\}_{i=1}^N$ satisfy
\[
\max_{i, j} \Big( \sin\theta_{ij}^{in} \Big)  <\frac{1}{2+\sqrt{2}\|W\|_F} \quad \mbox{and} \quad \max_{i,j} \theta_{ij}^{in} < \frac{\pi}{2},
\]
practical aggregation emerges asymptotically (see Theorem \ref{T3.1}):
\[
\sup_{t\geq0} \Big( \max_{i,j} \theta_{ij} \Big) \leq \frac{\pi}{2} \quad \mbox{and} \quad \lim_{\kappa\to\infty}\limsup_{t\to\infty} \max_{i, j} \Big(\sin\theta_{ij}(t) \Big)=0.
\]
Note that in \cite{H-K-P-R}, emergent dynamics for non-identical ensemble has not been studied. 

Third, we consider system \eqref{A-2-1} with $\kappa_1$:
\begin{equation*} \label{A-5}
\dot{z}_j =\Omega_jz_j+ \kappa_0 \Big( \langle z_j, z_j\rangle V_0 z_c - \langle z_c, z_j \rangle z_j \Big).
\end{equation*}
Next, we introduce real and imaginary parts of the two-point correlation function $\langle z_i, z_j \rangle$:
\begin{equation*} \label{A-6}
R_{ij} :=\mathrm{Re}(\langle z_i, z_j\rangle),\quad I_{ij} :=\mathrm{Im}(\langle z_i, z_j\rangle)\quad {\mathcal J}_{ij} :=\sqrt[4]{(1-R_{ij})^2+I_{ij}^2},
\end{equation*}
for $i, j\in\{1, 2, \cdots, N\}.$ For identical ensemble, if the coupling strength and  initial data satisfy
\[
\kappa_0 > 0 \quad \mbox{and} \quad \max_{i,j}{\mathcal J}_{ij}^{in} < \frac{2\sqrt{2}}{\sqrt{\sqrt{5} \|W_0\|_F^2 +8}+\sqrt[4]{5} \|W_0\|_F},
\]
then ${\mathcal J}_{ij}$ decays to zero exponentially fast, which illustrates the emergence of complete aggregation (see Theorem \ref{T4.1}). In contrast, for non-identical ensemble, if initial data $\{z_j^{in}\}$ satisfy
\[
\max_{i, j} \mathcal J_{ij}^{in}< \frac{2\sqrt{2}}{\sqrt{\sqrt{5} \|W_0\|_F^2 +8}+\sqrt[4]{5} \|W_0\|_F},
\]
then practical aggregation emerges, i.e., 
\[
\lim_{\kappa_0\to\infty}\limsup_{t\to\infty} \max_{i,j} {\mathcal J}_{ij}(t) =0,
\]
(See Theorem \ref{T4.2}). 

Lastly, we deal with the full dynamics \eqref{LHSF} with $\kappa_0 > 0$ and $\kappa_1 > 0$.  For identical ensemble, if coupling strengths, frustration matrix $W_1$ and initial data satisfy
\[
\kappa_0> 2 \kappa_1\geq0,  \quad W_1 \equiv 0, \quad  \max_{i,j} {\mathcal J}_{ij}^{in} < \frac{2\sqrt{2}\left( 1-\frac{2\kappa_1}{\kappa_0} \right)}{\sqrt{\sqrt{5} \|W_0\|_F^2 +8\left( 1-\frac{2\kappa_1}{\kappa_0} \right)}+\sqrt[4]{5} \|W_0\|_F}, 
\]
then, the complete aggregation emerges exponentially fast (see Theorem \ref{T5.1}).  In contrast, for non-identical ensemble, if coupling strength $\kappa_1$ is fixed and initial data satisfy 
\[
\max_{i, j} {\mathcal J}_{ij}^{in} < \frac{2\sqrt{2}}{\sqrt{\sqrt{5} \|W_0\|_F^2 +8}+\sqrt[4]{5} \|W_0\|_F},
\]
then practical aggregation emerges (see Theorem \ref{T5.2}):
\[
\lim_{\kappa_0\to\infty}\limsup_{t\to\infty}  \max_{i,j} \mathcal J_{ij}(t)=0.
\]

\vspace{0.5cm}

The rest of this paper is organized as follows. In Section \ref{sec:2}, we briefly review basic properties of the Lohe sphere and Lohe hermitian sphere models with frustrations, and recall previous results on the emergent dynamics of the aforementioned models with frustration. In Section \ref{sec:3}, we study emergent dynamics of the Lohe sphere model. In particular, our practical aggregation estimate improves earlier results. In Section \ref{sec:4}, we present emergent dynamics of the Lohe hermitian sphere model with $\kappa_1 = 0$. This is exactly complex analogue of the Lohe sphere model. The complex nature of ambient space will appear in conditions for complete and practical aggregation estimates. In Section \ref{sec:5}, we study emergent dynamics of the full dynamics \eqref{LHSF} and provide sufficient conditions leading to the complete and practical aggregations. In Section \ref{sec:6}, we provide several numerical examples and compare them with analytical results in previous sections. Finally, Section \ref{sec:7} is devoted to a brief summary of our main results and some remaining issues which were not discussed in this work. 

\section{Preliminaries} \label{sec:2}
\setcounter{equation}{0} 
In this section, we study basic properties of the Lohe hermitian sphere model with frustration, and briefly review earlier results on the emergent dynamics of Lohe type models with frustration such as ``{\it the Kuramoto model}" on the unit circle and ``{\it the Lohe sphere model}" on the unit sphere.
\subsection{The LHS model with frustration} \label{sec:2.1}
Consider the LHS model on ${\mathbb H}{\mathbb S}^d$ with frustration:~for $j = 1, \cdots, N$, 
\begin{equation} \label{NNN-1}
\dot{z}_j =\Omega_jz_j+\frac{\kappa_0}{N}\sum_{k=1}^N(\langle z_j, z_j\rangle V_0z_k-\langle V_0z_k, z_j\rangle z_j) +\frac{\kappa_1}{N}\sum_{k=1}^N(\langle z_j, V_1z_k\rangle-\langle V_1z_k, z_j\rangle )z_j.
\end{equation}
In order to rewrite system \eqref{NNN-1} into a mean-field form, we introduce a centroid $z_c := \frac{1}{N} \sum_{k=1}^{N} z_k$. Then, the LHS model \eqref{NNN-1} can be rewritten as a mean-field form:
\begin{equation} \label{B-0-1}
\dot{z}_j=\Omega z_j+\kappa_0 \Big(\langle{z_j, z_j}\rangle V_0 z_c-\langle{V_0 z_c, z_j}\rangle z_j \Big) +\kappa_1 \Big(\langle{z_j, V_1 z_c}\rangle-\langle{V_1 z_c, z_j}\rangle \Big)z_j.
\end{equation}
\begin{lemma} \label{L2.1}
Let $\{ z_j \}$ be a solution to \eqref{B-0-1}. Then $\|z_j \|$ is a conserved quantity:
\[ \frac{d}{dt} \| z_j \|  = 0, \quad \mbox{for all $t > 0$},~~j = 1, \cdots, N. \]
\end{lemma}
\begin{proof} Note that 
\begin{equation} \label{B-0-2}
 \frac{d}{dt} \|z_j \|^2 = \langle {\dot z}_j, z_j \rangle + \langle z_j, {\dot z}_j \rangle.
\end{equation} 
We use \eqref{B-0-1} to estimate the second term in \eqref{B-0-2}:
\begin{align}
\begin{aligned} \label{B-0-3}
\langle{z_j, \dot{z}_j}\rangle &= \Big \langle{z_j, \Omega_j z_j+\kappa_0 \Big(\langle{z_j, z_j}\rangle V_0 z_c-\langle{V_0 z_c, z_j}\rangle z_j \Big)+\kappa_1 \Big (\langle{z_j, V_1 z_c}\rangle-\langle{V_1 z_c, z_j}\rangle \Big)z_j} \Big \rangle\\
&=\langle{z_j, \Omega z_j}\rangle +  \kappa_1 \Big( -\langle{V_1 z_c, z_j} \rangle + \langle z_j, V_1 z_c \rangle   \Big) \| z_j \|^2.
\end{aligned}
\end{align}
Now, we use the relation $ \langle {\dot z}_j, z_j \rangle = \overline{  \langle z_j, {\dot z}_j \rangle}$ to see
\begin{equation} \label{B-0-4}
\langle {\dot z}_j, z_j \rangle = \overline{\langle z_j, {\dot z}_j \rangle} = \langle \Omega z_j, z_j \rangle + \kappa_1 \Big(  \langle z_j, V_1 z_c \rangle - \langle V_1 z_c, z_j \rangle   \Big)\| z_j \|^2.
\end{equation}
In \eqref{B-0-2}, we combine estimates \eqref{B-0-3} and \eqref{B-0-4} to obtain
\[ \label{B-0-5}
 \frac{d}{dt} \|z_j \|^2  = \langle{z_j, \dot{z}_j}\rangle + \langle {\dot z}_j, z_j \rangle = \langle{z_j, \Omega z_j}\rangle +  \langle \Omega z_j, z_j \rangle = \langle (\Omega^\dagger + \Omega) z_j, z_j \rangle = 0,
\]
which yields the desired estimate.
\end{proof}

\vspace{0.2cm}

Next, we study solution splitting property, when the system has the same free flow $\Omega_j = \Omega$:
\begin{equation} \label{NNN-2}
\dot{z}_j =\Omega z_j+\frac{\kappa_0}{N}\sum_{k=1}^N(\langle z_j, z_j\rangle V_0z_k-\langle V_0z_k, z_j\rangle z_j) +\frac{\kappa_1}{N}\sum_{k=1}^N(\langle z_j, V_1z_k\rangle-\langle V_1z_k, z_j\rangle )z_j.
\end{equation}
Now, we consider the associated linear and nonlinear flows:
\begin{equation} 
\begin{cases}  \label{B-0-6}
\displaystyle \dot{f}_j =\Omega f_j, \\
\displaystyle \dot{w}_j = \frac{\kappa_0}{N}\sum_{k=1}^N(\langle w_j, w_j\rangle \tilde{V_0} w_k-\langle \tilde{V_0} w_k, w_j\rangle w_j)  
+\frac{\kappa_1}{N}\sum_{k=1}^N(\langle w_j, \tilde{V_1} w_k\rangle-\langle \tilde{V_1}w_k, w_j\rangle ) w_j,
\end{cases}
\end{equation}
where
\begin{equation} \label{B-0-7}
\tilde{V}_0(t)=e^{-\Omega t} V_0 e^{\Omega t} \quad \mbox{and} \quad \tilde{V}_1(t)=e^{-\Omega t} V_1 e^{\Omega t}.
\end{equation}
Then, it is easy to see that $f_j(t) = e^{\Omega t} f_j^{in}$, and let ${\mathcal N} = {\mathcal N}(t)$ be solution operators to \eqref{B-0-6} such that 
\[ W(t) = {\mathcal N}(t) W^{in}, \quad W(t) = (w_1(t), \cdots, w_N(t)). \]
In next proposition, we study a solution splitting property of \eqref{NNN-2}.
\begin{proposition} \label{P2.1}
Let $\{z_j \}$ be a solution to \eqref{NNN-2} with initial data $z^{in}$. Then, one has 
\[ z_j(t) = e^{\Omega t} ({\mathcal N}(t) z^{in})_j, \quad t \geq 0. \]
\end{proposition}
\begin{proof}
We substitute the ansatz
\[
z_j(t)=e^{\Omega t} w_j(t),
\]
into system \eqref{B-0-1} to get 
\begin{align*}
\begin{aligned}
e^{\Omega t}(\dot{w}_j+\Omega w_j)=&\Omega e^{\Omega t} w_j+\kappa_0(\langle{e^{\Omega t}w_j, e^{\Omega t}w_j}\rangle V_0 e^{\Omega t} w_c-\langle{V_0 e^{\Omega t} w_c, e^{\Omega t} w_j}\rangle e^{\Omega t}w_j)\\
&+\kappa_1(\langle{e^{\Omega t}w_j, V_1 e^{\Omega t}w_c}\rangle-\langle{V_1 e^{\Omega t}w_c, e^{\Omega t} w_j}\rangle )e^{\Omega t} w_j.
\end{aligned}
\end{align*}
This leads
\begin{align}
\begin{aligned} \label{B-0-8}
\dot{w}_j &=\kappa_0 \Big( \langle{w_j, w_j}\rangle e^{-\Omega t}V_0 e^{\Omega t} w_c- \Big \langle{e^{-\Omega t}V_0 e^{\Omega t} w_c, w_j} \Big \rangle w_j \Big) \\
 &\hspace{2cm} +\kappa_1 \Big( \Big \langle{w_j, e^{-\Omega t}V_1  e^{\Omega t}w_c} \Big \rangle- \Big \langle{e^{-\Omega t}V_1  e^{\Omega t} w_c, w_j} \Big \rangle \Big )w_j.
\end{aligned}
\end{align}
Now we use \eqref{B-0-7} to simplify \eqref{B-0-8} to find
\[
\dot{w}_j=\kappa_0(\langle{w_j, w_j}\rangle \tilde{V}_0 w_c-\langle{\tilde{V}_0 w_c, w_j}\rangle w_j)+\kappa_1(\langle{w_j, \tilde{V}_1w_c}\rangle-\langle{\tilde{V}_1 w_c, w_j}\rangle )w_j.
\]
This complete the desired proof.
\end{proof}

\begin{remark}
Suppose that $W_0$ and $W_1$ satisfy
\begin{align*}
[W_0, \Omega] = 0 \quad \mbox{and} \quad [W_1, \Omega] = 0.
\end{align*}
Then, we have
\begin{align*}
\tilde{V}_0 = e^{-\Omega t} V_0 e^{\Omega t} = V_0 \quad \mbox{and} \quad \tilde{V}_1 = e^{-\Omega t} V_1 e^{\Omega t} = V_1,
\end{align*}
which implies the solution splitting property.
\end{remark}

Now we will study about relations between the Lohe hermitian sphere model with frustration and other synchronization models with frustration. First, if we set initial data $Z^{in}=\{z_j^{in}\}_{j=1}^N$ as set of real unit vectors and $\Omega_j, W_0$ as real skew-symmetric matrix, then we can easily obtain the Lohe sphere model with frustration. We introduce how we can reduce the Lohe hermitian sphere model with frustration to Kuramoto model with frustration. The diagram below shows relations among the Lohe hermitian sphere model with frustration, the Lohe sphere model with frustration, and the Kuramoto model with frustration. Since we can obtain the Lohe sphere model with frustration from the Lohe hermitian sphere model with frustration by letting $\kappa_1 = 0$, we denoted it in the diagram. 

\begin{equation*}
\begin{tikzcd}
    \thead{\text{Lohe Hermitian sphere model}\\\text{with frustration}} \arrow{rd}{\kappa_0=0,~d=0} \arrow{r}{\kappa_1=0}
    & \thead{\text{Lohe sphere model}\\\text{with frustration}}          \arrow{d}{d=1}\\
    &\thead{\text{Kuramoto model}\\\text{with frustration}}
\end{tikzcd}
\end{equation*}
Next, we show how the LS model with frustration can be reduced to the Kuramoto model with frustration \cite{Da, H-K-L, H-K-L1, L-H, Zh}. In what follows, we consider two subsystems of \eqref{LHSF} separately. \newline

\noindent $\bullet$~Subsystem A ($\kappa_0 > 0$ and $\kappa_1 = 0$):~Consider the LHS model restricted on $\bbs^d$ with a uniform frustration:
\begin{equation} \label{B-0-9}
\dot{z}_i=\Omega_i z_i+\frac{\kappa}{N}\sum_{k=1}^N\Big(V_0z_k-\langle z_i, V_0 z_k\rangle z_i \Big),
\end{equation}
where $V_0$ is the frustration matrix of the form \eqref{fr}. For $d= 1$, We set 
\begin{equation} \label{B-0-10}
z_i :=\begin{bmatrix}
\cos\theta_i\\
\sin\theta_i
\end{bmatrix}, \quad 
\Omega_i :=\begin{bmatrix}
0&-\nu_i\\
\nu_i&0
\end{bmatrix}, \quad 
V_0  :=\begin{bmatrix}
\cos\alpha&-\sin\alpha\\
\sin\alpha&\cos\alpha
\end{bmatrix}.
\end{equation}
We substitute \eqref{B-0-10} into \eqref{B-0-9} to obtain
\begin{align*}
\begin{aligned}
&\dot{\theta}_i\begin{bmatrix}
-\sin\theta_i\\
\cos\theta_i
\end{bmatrix} =\begin{bmatrix}
0&-\nu_i\\
\nu_i&0
\end{bmatrix}\begin{bmatrix}
\cos\theta_i\\
\sin\theta_i
\end{bmatrix} \\
&+\frac{\kappa}{N}\sum_{k=1}^N
\left(\begin{bmatrix}
\cos\alpha&-\sin\alpha\\
\sin\alpha&\cos\alpha
\end{bmatrix}\begin{bmatrix}
\cos\theta_k\\
\sin\theta_k
\end{bmatrix}
-
\left\langle \begin{bmatrix}
\cos\theta_i\\
\sin\theta_i
\end{bmatrix}, \begin{bmatrix}
\cos\alpha&-\sin\alpha\\
\sin\alpha&\cos\alpha
\end{bmatrix}\begin{bmatrix}
\cos\theta_k\\
\sin\theta_k
\end{bmatrix} \right\rangle\begin{bmatrix}
\cos\theta_i\\
\sin\theta_i
\end{bmatrix}\right) \\
&=\nu_i\begin{bmatrix}
-\sin\theta_i\\
\cos\theta_i
\end{bmatrix}+\frac{\kappa}{N}\sum_{k=1}^N
\begin{bmatrix}
\cos\alpha\cos\theta_k-\sin\alpha\sin\theta_k-(\cos\theta_i\cos(\theta_k+\alpha)+\sin\theta_i\sin(\theta_k+\alpha))\cos\theta_i\\
\sin\alpha\cos\theta_k+\cos\alpha\sin\theta_k-(\cos\theta_i\cos(\theta_k+\alpha)+\sin\theta_i\sin(\theta_k+\alpha))\sin\theta_i
\end{bmatrix}\\
&=\nu_i\begin{bmatrix}
-\sin\theta_i\\
\cos\theta_i
\end{bmatrix}+\frac{\kappa}{N}\sum_{k=1}^N
\begin{bmatrix}
\cos(\theta_k+\alpha)-\cos(\theta_i-\theta_k-\alpha)\cos\theta_i\\
\sin(\theta_k+\alpha)-\cos(\theta_i-\theta_k-\alpha)\sin\theta_i
\end{bmatrix}\\
&=\nu_i\begin{bmatrix}
-\sin\theta_i\\
\cos\theta_i
\end{bmatrix}+\frac{\kappa}{N}\sum_{k=1}^N
\begin{bmatrix}
-\sin\theta_i\sin(\theta_k-\theta_i+\alpha)\\
\cos\theta_i\sin(\theta_k-\theta_i+\alpha)
\end{bmatrix}.
\end{aligned}
\end{align*}
This leads to the Kuramoto model with frustration:
\[
\dot{\theta}_i=\nu_i+\frac{\kappa}{N}\sum_{k=1}^N\sin(\theta_k-\theta_i+\alpha).
\]
For $\alpha = 0$, the above system is exactly the Kuramoto model \cite{C-H-J-K, C-S, D-B1, H-K-R, Ku2}.

\vspace{0.5cm}

\noindent $\bullet$~Subsystem B ($\kappa_0 = 0$ and $\kappa_1 > 0$):~Consider the Subsystem B:
\begin{equation} \label{B-0-11}
\dot{z}_j=\Omega_jz_j+\frac{\kappa_1}{N}\sum_{k=1}^N\left(\langle z_j, V_1 z_k\rangle-\langle V_1z_k, z_j\rangle\right)z_j.
\end{equation}
We set 
\[
d=0,\quad \Omega_j=\nu_j\mathrm{i},\quad z_j=e^{\mathrm{i}\theta_j},\quad V_1=e^{\mathrm{i}\alpha},\quad \kappa_1=\frac{\kappa}{2}.
\]
We substitute the above ansatz into \eqref{B-0-11} to find
\begin{align*}
\mathrm{i}e^{\mathrm{i}\theta_j}\dot{\theta}_j&=\nu_j\mathrm{i}e^{\mathrm{i}\theta_j}+\frac{\kappa}{2N}\sum_{k=1}^N
\left(
e^{\mathrm{i}(\theta_k+\alpha-\theta_j)}-e^{\mathrm{i}(\theta_j-\theta_k-\alpha)}
\right)e^{\mathrm{i}\theta_j}\\
&=\mathrm{i}e^{\mathrm{i}\theta_j}\left(\nu_j +\frac{\kappa}{N}\sum_{k=1}^N\sin(\theta_k-\theta_j+\alpha)\right).
\end{align*}
After simplification, we obtain the Kuramoto model with frustration:
\[
\dot{\theta}_j=\nu_j+\frac{\kappa}{N}\sum_{k=1}^N\sin(\theta_k-\theta_j+\alpha).
\]

\subsection{Previous results} \label{sec:2.3}
 In this subsection, we briefly recall previous results on the emergent dynamics of two aforementioned particle models with frustration.

 \subsubsection{The Kuramoto model}  \label{sec:2.3.1}
 Let $\theta_j = \theta_j(t)$ be the phase of the $j$-th oscillator whose dynamics is governed by the Kuramoto model with frustration:
\begin{equation} \label{B-1}
{\dot \theta}_j = \nu_j + \frac{\kappa}{N} \sum_{k=1}^{N} \sin(\theta_k - \theta_j + \alpha), \quad |\alpha| < \frac{\pi}{2},
\end{equation}
where $\alpha$ is the uniform size of frustration. \newline

For a brief description of \cite{H-K-L}, we set
\[ D(\Theta) := \max_{i,j} |\theta_i - \theta_j|, \quad D({\dot \Theta}) := \max_{i,j} |{\dot \theta}_i - {\dot \theta}_j|, \quad D(\nu) := \max_{i,j} |\nu_i - \nu_j |.  \]
Note that $D(\Theta(t))$ is Lipschitz continuous, hence it is differentiable except at times of collision between the extremal phases and their neighboring phases. \newline

Let $\kappa, D(\nu)$ and $\alpha$ be positive constants satisfying 
\[ D(\nu) + \kappa \sin |\alpha| < \kappa. \]
Then we set $D^{\infty}_{1} < D^{\infty}_2$ be two roots of the following trigonometric equation:
\[ \displaystyle \sin x = \frac{D(\nu) + \kappa \sin |\alpha|}{\kappa},\quad x \in \left(0, \pi \right). \]
\begin{proposition}\cite{H-K-L} \label{P2.2}
Suppose system parameters $D(\nu), \kappa$ and $\alpha$ satisfy
\[
D(\nu) > 0, \quad \kappa \geq \frac{D(\nu)}{1 - \sin |\alpha|}, \quad 0 < D(\Theta^0) < D^\infty_2 - |\alpha|,
\]
and let $\Theta = \Theta(t)$ be a solution to \eqref{B-1}. Then there exists $t_0 > 0$ such that
\[
D({\dot \Theta}(t_{0})) e^{-\kappa(t-t_{0})} \leq D({\dot \Theta}(t)) \leq D({\dot \Theta}(t_{0})) e^{-\kappa \cos (D_1^{\infty} + \varepsilon)(t-t_{0})}, \quad t \geq t_0,
\]
where $\varepsilon\ll 1$ is a positive constant satisfying $D_1^{\infty} +\varepsilon<\frac{\pi}{2}$.
\end{proposition}

\begin{remark}\label{R2.2} 
For the Kuramto model with heterogeneous frustrations with $\alpha_{kj}$, its emergent dynamics has been studied in \cite{H-K-P}.
\end{remark}

\subsubsection{The LS model with the same free flow} \label{sec:2.2.2} Let $x_j = x_j(t)$ be a position of the $j$-th particle on the unit sphere ${\mathbb S}^d$ whose dynamics is governed by the following system with frustration and the same free flow:
\begin{equation} \label{B-2}
\dot x_j = \Omega x_j + \frac{\kappa}{N} \sum_{k=1}^N  (Vx_k - \langle x_j, Vx_k\rangle x_j), 
\end{equation}
where the frustration matrix $V$ is given by the following form:
\begin{equation} \label{B-3}
 V = I_{d+1} + W, 
 \end{equation}
where  $I_{d+1}$ is the $(d+1)\times (d+1)$ identity matrix and $ W$ is a $(d+1)\times (d+1)$ skew-symmetric matrix. \newline

We further substitute \eqref{B-3} into \eqref{B-2} to get 
\begin{equation} \label{B-3-1}
\dot x_j = \Omega x_j + \underbrace{\frac{\kappa}{N} \sum_{k=1}^N (x_k - \langle x_j,x_k\rangle x_j)}_\textup{synchronous motion} + \underbrace{ \frac{\kappa}{N} \sum_{k=1}^N ( Wx_k - \langle x_j, Wx_k\rangle x_j ). }_\textup{ periodic motion} 
\end{equation}
Note that the presence of frustration matrix $W$ induces a competition between \textit{`synchronization'} and \textit{`periodic motion'}, in the following sense. The second term on the RHS of \eqref{B-3-1} tends to bring the oscillators together. On the other hand, since $W$ is a $(d+1) \times (d+1)$ skew-symmetric matrix,  all eigenvalues of $W$ are zero or purely imaginary. Hence, we can interpret the last term on the RHS of \eqref{B-3-1}, together with $\Omega x_j$, tries to pull the dynamics into a periodic motion.
  
\begin{theorem} \label{T2.1}
Suppose natural frequency matrix, frustration matrix and initial data satisfy 
\[ \label{Z-22-1}
 \| W \|_F < 1, \quad  \max_{i,j} \Big(1-\langle x_i^{in},x_j^{in} \rangle \Big) < 1- \|W\|_F, \]
where $W$ is a $(d+1)\times(d+1)$ skew-symmetric matrix, and let $\{x_i\}_{i=1}^N$ be a solution to \eqref{B-2}. Then, one has 
\[  \lim_{\kappa \to \infty} \limsup_{t\to\infty} \max_{i,j} \|x_i(t) -x_j(t)\|  = 0. \]
\end{theorem}

\vspace{0.5cm}

In next three sections, we study the emergent dynamics of \eqref{LHSF} with heterogeneous frustrations into two cases. First, we study the emergence of practical aggregation for the case with $\kappa_1 = 0$. Second, we study the practical aggregation of \eqref{LHSF} with $\kappa_0 > 0$ and $\kappa_1 > 0$.  More precisely, we study the LH model with frustration, i.e., either
\[
\kappa_0 > 0, \quad \kappa_1 = 0, \quad \mbox{or} \quad \kappa_0 > 0, \quad \kappa_1 > 0.
\]
Consider the LHS model with $\kappa_1 = 0$:
\[
\mbox{Subsystem A:}~ \dot{z}_j=\Omega_j z_j+\kappa_0\Big(\langle{z_j, z_j}\rangle V_0 z_c-\langle{V_0 z_c, z_j}\rangle z_j\Big)
\]
and 
\[
\mbox{Subsystem B:}~\dot{z}_j= \Omega_j z_j+\kappa_0\Big(\langle{z_j, z_j}\rangle V_0 z_c-\langle{V_0 z_c, z_j}\rangle z_j\Big) + \kappa_1\Big(\langle{z_j, V_1 z_c}\rangle-\langle{V_1 z_c, z_j}\rangle\Big)z_j.
\]
For Subsystem A, it is easy to see that 
\[  z_i^{in} \in \bbr^{d+1}  \quad \Longrightarrow \quad z_j(t) \in \bbr^{d+1}, \quad t > 0. \]

\section{Emergent dynamics of the LS model:~ real vector case} \label{sec:3}
\setcounter{equation}{0}
In this section, we provide improved aggregation estimate for the LHS model with frustration which can be obtained from the LHS model restricted on the unit sphere $\bbs^d$. In this case, the state $z_j = x_j  \in \bbr^d$. Then, the state $x_j$ satisfies 
\begin{equation}\label{C-0-0}
\mbox{System A:}~ \dot{x}_j=\Omega_j x_j+\kappa\Big(\langle{x_j, x_j}\rangle V x_c-\langle{V x_c, x_j}\rangle x_j \Big),
\end{equation}
where $V=I_{d+1}+W$.
\newline

Before we discuss the aggregation estimate, we recall the concept of ``{\it complete and practical aggregations}".  
\begin{definition} \label{D3.1}
Let ${\mathcal X} = \{x_j \}$ be a dynamic ensemble whose evolution is governed by system \eqref{C-0-0}. 
\begin{enumerate}
\item
The ensemble ${\mathcal X}$ exhibits complete aggregation, if the following estimate holds.
\[   \lim_{t \to \infty} D({\mathcal X}(t)) = 0, \]
where $D({\mathcal X}):=\max_{i,j}\|x_i-x_j\|$.
\vspace{0.1cm}
\item
The ensemble ${\mathcal X}$  exhibits practical aggregation, if the following estimate holds.
\[  \lim_{\kappa \to \infty} \limsup_{t \to \infty} D({\mathcal X}(t)) = 0.   \] 
\end{enumerate}
\end{definition}

\vspace{0.2cm}

For later practical aggregation estimate, the following lemma will be used crucially. 
\begin{lemma}\label{L3.1}
Let $W\in \bbr^{d \times d}$ be a real skew-symmetric matrix and $x,y\in\bbr^{d}$ be vectors. Then we have 
\[
\Big|\langle x, Wy\rangle \Big|\leq \frac{1}{\sqrt{2}}\|W\|_F\sqrt{\|x\|^2\|y\|^2-\langle x, y\rangle^2},
\]
where $\langle \cdot , \cdot \rangle $is the inner product in $\bbr^{d}$. 
\end{lemma}
\begin{proof} 
 By direct calculations, one has 
\begin{align*}
\langle x, Wy\rangle&=\sum_{i, j}[x]_i[W]_{ij}[y]_j =\frac{1}{2}\sum_{i, j}\left([x]_i[W]_{ij}[y]_j+[x]_j[W]_{ji}[y]_i\right) \\
&=\frac{1}{2}\sum_{i, j}\left([x]_i[W]_{ij}[y]_j-[x]_j[W]_{ij}[y]_i\right) =\frac{1}{2}\sum_{i, j}[W]_{ij}([x]_i[y]_j-[y]_i[x]_j),
\end{align*}
where we used dummy variable exchange and skew-symmetry of matrix $W$.\newline

Then, it follows from the Cauchy-Schwarz inequality that 
\begin{align*}
|\langle x, Wy\rangle|^2&=\frac{1}{4}\left|\sum_{i, j}[W]_{ij}([x]_i[y]_j-[y]_i[x]_j)\right|^2 \leq \frac{1}{4} \left(\sum_{i, j}[W]_{ij}^2\right)\cdot\left(\sum_{i, j}([x]_i[y]_j-[y]_i[x]_j)^2\right)\\
&=\frac{1}{4}\|W\|_F^2 \sum_{i, j}([x]_i^2[y]_j^2+[y]_i^2[x]_j^2-2[x]_i[y]_j[y]_i[x]_j)\\
&=\frac{1}{2}\|W\|_F^2\cdot(\|x\|^2\|y\|^2-\langle x, y\rangle^2).
\end{align*}
This yields the desired estimate.
\end{proof}

\vspace{0.5cm}

For emergent dynamics, we introduce $\theta_{ij}$ measuring the angle between $x_i$ and $x_j$:
\begin{equation} \label{New-1}
\theta_{ij} :=\cos^{-1}\left(\langle x_i, x_j\rangle\right), \quad 1 \leq i, j \leq N.
\end{equation}
\begin{lemma} \label{L3.2}
Let $\{x_j \}$ be a solution to \eqref{C-0-0}. Then, $\theta_{ij}$ satisfies
\begin{align*}
\dot{\theta}_{ij} \leq  \frac{1}{2}\|\Omega_i-\Omega_j\|_F
 -\frac{\kappa}{N}\tan\left(\frac{\theta_{ij}}{2}\right)\sum_{k=1}^N \Big[ \cos\theta_{ik}+\cos\theta_{jk}-\frac{\|W\|_F}{\sqrt{2}}\Big (\sin\theta_{ik}+\sin\theta_{jk} \Big) \Big].
\end{align*}
\end{lemma}
\begin{proof}
It follows from \eqref{C-0-0} that 
\begin{align*}
\frac{d}{dt}\langle x_i, x_j\rangle &=\left\langle \Omega_ix_i+\frac{\kappa}{N}\sum_{k=1}^N(V_0x_k-\langle V_0x_k, x_i\rangle x_i), x_j\right\rangle \\
&+\left\langle x_i, \Omega_jx_j+\frac{\kappa}{N}\sum_{k=1}^N(V_0x_k-\langle V_0x_k, x_j\rangle x_j)\right\rangle\\
&=\langle \Omega_ix_i, x_j\rangle+\langle x_i, \Omega_j x_j\rangle+\frac{\kappa}{N}\sum_{k=1}^N\big(\langle V_0x_k, x_j\rangle+\langle x_i, V_0x_k\rangle\big)(1-\langle x_i, x_j\rangle).
\end{align*}
Now, we use the skew-symmetry of $\Omega_i, \Omega_j$ and $W$ to find
\begin{align*}
\begin{aligned}
\frac{d}{dt}\langle x_i, x_j\rangle
&=\langle (\Omega_i-\Omega_j)x_i, x_j\rangle+(1-\langle x_i, x_j\rangle)\cdot \frac{\kappa}{N}\sum_{k=1}^N\langle x_k, x_i+x_j\rangle \\
& +(1-\langle x_i, x_j\rangle)\cdot\frac{\kappa}{N}\sum_{k=1}^N\langle 
Wx_k, x_i+x_j\rangle.
\end{aligned}
\end{align*}
We use the defining relation \eqref{New-1} for $\theta_{ij}$ to obtain
\begin{align*}
\begin{aligned}
\frac{d}{dt}\cos\theta_{ij} &= \Big \langle (\Omega_i-\Omega_j)x_i, x_j \Big \rangle+\frac{\kappa}{N}(1-\cos\theta_{ij})\sum_{k=1}^N(\cos\theta_{ik}+\cos\theta_{jk}) \\
&\hspace{0.5cm} +\frac{\kappa}{N}(1-\cos\theta_{ij})\sum_{k=1}^N \Big(\langle Wx_k, x_i\rangle+\langle Wx_k, x_j\rangle \Big),
\end{aligned}
\end{align*}
or equivalently
\begin{align}
\begin{aligned}\label{C-0}
\dot{\theta}_{ij}&=-\underbrace{\frac{\langle(\Omega_i-\Omega_j)x_i, x_j\rangle}{\sin\theta_{ij}}}_{=:\mathcal{I}_{11}}-\underbrace{\frac{\kappa}{N}\frac{1-\cos\theta_{ij}}{\sin\theta_{ij}}\sum_{k=1}^N(\cos\theta_{ik}+\cos\theta_{jk})}_{=:\mathcal{I}_{12}}\\
&\hspace{3cm}-\underbrace{\frac{\kappa}{N}\frac{1-\cos\theta_{ij}}{\sin\theta_{ij}}\sum_{k=1}^N(\langle Wx_k, x_i\rangle+\langle Wx_k, x_j\rangle)}_{=:\mathcal{I}_{13}}.
\end{aligned}
\end{align}
\vspace{0.2cm}

\noindent $\bullet$~(Estimate of ${\mathcal I}_{11}$): We use Lemma \ref{L3.1} with $A:= \Omega_i - \Omega_j$ which is a skew-symmeric to obtain
\[
|\langle (\Omega_i - \Omega_j) x_i, x_j\rangle|\leq \frac{1}{\sqrt{2}}\|\Omega_i - \Omega_j \|_F\sqrt{\|x_i\|^2\|x_j\|^2-\langle x_i, x_j\rangle^2}=\frac{1}{\sqrt{2}}\|\Omega_i - \Omega_j \|_F\sin\theta_{ij}.
\]
Then we use the above estimate to derive 
\begin{equation} \label{New-B}
|\mathcal{I}_{11}|\leq \frac{1}{\sqrt{2}}\|\Omega_i-\Omega_j\|_F.
\end{equation}

\vspace{0.2cm}

\noindent $\bullet$~(Estimate of ${\mathcal I}_{12}$ and ${\mathcal I}_{13}$): Similarly, one has 
\begin{align}
\begin{aligned}  \label{New-C}
|\mathcal{I}_{12}| &\leq\frac{\kappa}{N}\tan\left(\frac{\theta_{ij}}{2}\right)\sum_{k=1}^N(\cos\theta_{ik}+\cos\theta_{jk}), \\
|\mathcal{I}_{13}| &\leq \frac{\kappa\|W\|_F}{\sqrt{2}N}\tan\left(\frac{\theta_{ij}}{2}\right)\sum_{k=1}^N(\sin\theta_{ik}+\sin\theta_{jk}).
\end{aligned}
\end{align}
In \eqref{C-0}, we combine all the estimates \eqref{New-B} and \eqref{New-C} to get 
\begin{align}
\begin{aligned}\label{C-0-1}
\dot{\theta}_{ij}&\leq \frac{1}{2}\|\Omega_i-\Omega_j\|_F-\frac{\kappa}{N}\tan\left(\frac{\theta_{ij}}{2}\right)\sum_{k=1}^N(\cos\theta_{ik}+\cos\theta_{jk}) \\
& +\frac{\kappa\|W\|_F}{\sqrt{2}N}\tan\left(\frac{\theta_{ij}}{2}\right)\sum_{k=1}^N(\sin\theta_{ik}+\sin\theta_{jk})\\
&=\frac{1}{2}\|\Omega_i-\Omega_j\|_F-\frac{\kappa}{N}\tan\left(\frac{\theta_{ij}}{2}\right)\sum_{k=1}^N \Big[ \cos\theta_{ik}+\cos\theta_{jk}-\frac{\|W\|_F}{\sqrt{2}} \Big(\sin\theta_{ik}+\sin\theta_{jk} \Big) \Big].
\end{aligned}
\end{align}
\end{proof}
Now, we are ready to provide our first main result on the aggregation of pairwise particles. 
\begin{proposition}\label{P3.1}
\emph{(Complete aggregation)}
Suppose the initial data $\Theta^{in}$ satisfy
\begin{align} \label{C-0-2}
\theta_{ij}^{in} < \cot^{-1} \left( \frac{\|W\|_F}{\sqrt{2}} \right), \quad \forall i, j\in\{1, 2, \cdots, N\},
\end{align}
and let $\{x_i\}$ be the solution of \eqref{C-0} with the same free flow $\Omega_j \equiv \Omega$. Then  there exists a positive constant $\Lambda_{ij} = \Lambda_{ij}(N, \kappa, W, \Theta^{in})$ such that 
\[
\theta_{ij}(t)\leq \theta_{ij}^{in}\exp\left(-\Lambda_{ij}t\right), \quad t \geq 0.
\]
\end{proposition}
\begin{proof}
We use $\Omega_j = \Omega$ and  Lemma \ref{L3.2} to find
\begin{align*}
\dot{\theta}_{ij} \leq-\frac{\kappa}{N}\tan\left(\frac{\theta_{ij}}{2}\right)\sum_{k=1}^N\left(\cos\theta_{ik}+\cos\theta_{jk}-\frac{\|W\|_F}{\sqrt{2}}(\sin\theta_{ik}+\sin\theta_{jk})\right).
\end{align*}
On the other hand, it follows from \eqref{C-0-2} that 
\[
\theta_{ij}^{in} <\cot^{-1}\left(\frac{\|W\|_F}{\sqrt{2}}\right)\quad\Longrightarrow\quad \cos\theta_{ij}^{in}-\frac{\|W\|_F}{\sqrt{2}}\sin\theta_{ij}^{in} >0.
\]
This implies 
\[  \dot{\theta}_{ij} \Big|_{t = 0+} \leq 0, \quad \forall~i, j\in \{1, 2, \cdots, N\}. \]
Note that $f(\theta)=\cos\theta-\displaystyle\frac{\|W\|_F}{\sqrt{2}}\sin\theta$ is a decreasing function for $\theta\in[0, \pi]$. Thus, it is easy to see
\begin{equation} \label{C-0-2-1}
\dot{\theta}_{ij}\leq -\frac{\kappa}{N}\tan\left(\frac{\theta_{ij}}{2}\right)\sum_{k=1}^N\left(\cos\theta_{ik}^{in} +\cos\theta_{jk}^{in} -\frac{\|W\|_F}{\sqrt{2}}(\sin\theta_{ik}^{in} +\sin\theta_{jk}^{in})\right).
\end{equation}
Then, we use \eqref{C-0-2-1} and the relation $\tan\theta\geq\theta, \quad \theta \in [0, \pi/2)$ to find 
\[
\dot{\theta}_{ij}\leq -\frac{\kappa}{2N}\theta_{ij}\sum_{k=1}^N\left(\cos\theta_{ik}^{in}+\cos\theta_{jk}^{in} -\frac{\|W\|_F}{\sqrt{2}}(\sin\theta_{ik}^{in}+\sin\theta_{jk}^{in})\right)=: -\Lambda_{ij}\theta_{ij},
\]
which implies the desired estimate.
\end{proof}
Next, we consider the heterogenous ensemble in the sense that there exists $i \not = j$ such that 
\[ \Omega_i \not = \Omega_j. \]
In the following theorem, we consider the second main result on the emergence of practical aggregation. 
\begin{theorem}\label{T3.1}
\emph{(Practical aggregation)}
Let $\{x_i\}$ be the solution of \eqref{C-0} with initial data $\{x_i^{in}\}:$
\[
\max_{i, j} \Big( \sin\theta_{ij}^{in} \Big)  <\frac{1}{2+\sqrt{2}\|W\|_F} \quad \mbox{and} \quad \max_{i,j} \theta_{ij}^{in} < \frac{\pi}{2}.
\]
Then, we have
\[
\sup_{t\geq0}\Big( \max_{i,j}\theta_{ij} \Big) \leq \frac{\pi}{2} \quad \mbox{and} \quad \lim_{\kappa\to\infty}\limsup_{t\to\infty} \max_{i, j} \Big(\sin\theta_{ij}(t) \Big)=0.
\]

\end{theorem}
\begin{proof} For a distributed set $\{\Omega_j \}$ of natural frequency matrices,  we set
\[
\mathcal{D}(\Omega) :=\max_{i, j}\|\Omega_i-\Omega_j\|_F.
\]
Then, it follows from Lemma \ref{L3.2} that 
\begin{align*}
\dot{\theta}_{ij} \leq\frac{1}{2}\mathcal{D}(\Omega)-\frac{\kappa}{N}\tan \Big( \frac{\theta_{ij}}{2} \Big) \sum_{k=1}^N\Big[ \cos\theta_{ik}+\cos\theta_{jk}-\frac{\|W\|_F}{\sqrt{2}}\Big(\sin\theta_{ik}+\sin\theta_{jk} \Big) \Big].
\end{align*}
Now we set 
\[ \theta_M(t) :=\max_{i, j}\theta_{ij}(t), \quad {\mathcal T}:= \{t_0~:~\max_{i,j} \theta_{ij}(t) < \frac{\pi}{2}, \quad \forall~t \in [0, t_0) \}, \quad T^{\infty} := \sup {\mathcal T}.
\]
Then we have
\begin{align*}
\dot{\theta}_M&\leq \frac{1}{2}\mathcal{D}(\Omega)-\frac{\kappa}{N}\tan\left(\frac{\theta_M}{2}\right)\sum_{k=1}^N\Big [ \cos\theta_{ik}+\cos\theta_{jk}-\frac{\|W\|_F}{\sqrt{2}}\Big(\sin\theta_{ik}+\sin\theta_{jk} \Big) \Big] \\
&\leq\frac{1}{2}\mathcal{D}(\Omega)-\kappa\tan\left(\frac{\theta_M}{2}\right)\left(2\cos\theta_{M}-\sqrt{2}\|W\|_F\sin\theta_{M}\right), \quad \mbox{a.e.}~t\in[0, T^\infty).
\end{align*}
This yields,
\begin{align}
\begin{aligned} \label{New-4}
&\frac{1}{2}\cos\left(\frac{\theta_M}{2}\right)\dot{\theta}_M \leq \frac{1}{4}\mathcal{D}(\Omega)\cos\left(\frac{\theta_M}{2}\right) \\
& \hspace{0.5cm} -\kappa\sin\left(\frac{\theta_M}{2}\right)\Big[ 2\left(1-2\sin^2\left(\frac{\theta_M}{2}\right)\right)-2\sqrt{2}\|W\|_F\sin\left(\frac{\theta_M}{2}\right)\cos\left(\frac{\theta_M}{2}\right) \Big]\\
& \hspace{0.5cm} \leq \frac{1}{4}\mathcal{D}(\Omega)-\kappa\sin\left(\frac{\theta_M}{2}\right)\Big[ 2\left(1-2\sin\left(\frac{\theta_M}{2}\right)\right)-2\sqrt{2}\|W\|_F\sin\left(\frac{\theta_M}{2}\right) \Big] \\
& \hspace{0.5cm} \leq \frac{1}{4}\mathcal{D}(\Omega)-\kappa\sin\left(\frac{\theta_M}{2}\right)\Big[ 2-(4+2\sqrt{2}\|W\|_F)\sin\left(\frac{\theta_M}{2}\right) \Big], \quad \mbox{a.e.}~t\in[0, T^\infty).
\end{aligned}
\end{align}
To simplify \eqref{New-4} further, we set 
\[ s :=\sin\left(\frac{\theta_M}{2}\right). \]
Then, we have
\begin{align}\label{C-0-3}
\dot{s}\leq \frac{1}{4}\mathcal{D}(\Omega)-\kappa s\Big(2-(4+2\sqrt{2}\|W\|_F)s\Big), \quad  \quad 0 \leq t < T^{\infty}.
\end{align}

For a practical aggregation estimate, we uses a similar argument in \cite{H-N-P}. We define the following quadratic polynomial:
\[
p(s)=\frac{1}{4}\mathcal{D}(\Omega)-\kappa s(2-(4+2\sqrt{2}\|W\|_F)s)=(4+2\sqrt{2}\|W\|_F)\kappa s^2-2\kappa s+\frac{1}{4}\mathcal{D}(\Omega).
\]
Then the discriminant $D$ of above quadratic polynomial is
\[
D=4\kappa^2-\mathcal{D}(\Omega)\kappa(4+2\sqrt{2}\|W\|_F)=\kappa(4\kappa-\mathcal{D}(\Omega)(4+2\sqrt{2}\|W\|_F)).
\]
If
\[
\kappa>\frac{\mathcal{D}(\Omega)(2+2\sqrt{2}\|W\|_F)}{4},
\]
we have two distinct roots $s_1<s_2$:
\begin{align*}
\begin{aligned}
s_1 &=\frac{2\kappa-\sqrt{4\kappa^2-\kappa(4+2\sqrt{2}\|W\|_F)\mathcal{D}(\Omega)}}{2(4+2\sqrt{2}\|W\|_F)\kappa}, \\
s_2 &=\frac{2\kappa+\sqrt{4\kappa^2-\kappa(4+2\sqrt{2}\|W\|_F)\mathcal{D}(\Omega)}}{2(4+2\sqrt{2}\|W\|_F)\kappa}.
\end{aligned}
\end{align*}
Moreover, since the coefficient of $s^2$ and $p(0)$ are positive, one has 
\[ 0<s_1. \]
If $s(t)$ is the solution of \eqref{C-0-3} with initial data $s(0)<s_2$, we have
\[
T^\infty = \infty \quad \mbox{and} \quad \limsup_{t\to\infty}s(t)\leq s_1.
\]
Explicit formula of $s_1$ and $s_2$ provide us 
\[
\lim_{\kappa\to\infty}s_1=0 \quad \mbox{and} \quad \lim_{\kappa\to\infty}s_2=\frac{1}{2+\sqrt{2}\|W\|_F},
\]
so that we can conclude that if $s(0)<\displaystyle\frac{1}{2+\sqrt{2}\|W\|_F}$ we have
\[
\lim_{\kappa\to\infty}\limsup_{t\to\infty}s(t)=0.
\]
\end{proof}

\section{Emergent dynamics of the LS model: complex vector case} \label{sec:4}
\setcounter{equation}{0}
In this section, we study emergent dynamics of the complex LS model. First, we consider the case in which the second coupling is not present, i.e., $\kappa_1 = 0$, i.e.,
\begin{equation*} \label{D-0}
 \dot{z}_j=\Omega_j z_j+\kappa_0\Big(\langle{z_j, z_j}\rangle V_0 z_c-\langle{V_0 z_c, z_j}\rangle z_j \Big). 
\end{equation*} 
For this, we first generalize the result of Lemma \ref{L3.1} as follows.
\begin{lemma}\label{L4.1}
Let $x, y\in\bbc^d$ be complex vectors, and $W \in \bbc^{d \times d}$ be a skew-hermitian matrix i.e., $W^{\dagger} = -W$. Then, one has
\[
\Big|\langle Wx, y\rangle+\langle y, Wx\rangle \Big|\leq \sqrt{2}\|W\|_F\sqrt{\|x\|^2\|y\|^2-\mathrm{Re}(\langle x, y\rangle^2)}.
\]
\end{lemma}
\begin{proof}
We basically use the same argument as in Lemma \ref{L3.1}.
\begin{align*}
\begin{aligned}
&\langle Wx, y\rangle+\langle y, Wx\rangle \\
&\hspace{0.5cm} =\sum_{i, j}\left(\overline{[W]_{ij}[x]_j}[y]_i+\overline{[y]_i}[W]_{ij}[x]_j\right) =\sum_{i, j}\left(-\overline{[x]_j}[W]_{ji}[y]_i+\overline{[y]_i}[W]_{ij}[x]_j\right)\\
&\hspace{0.5cm} =\sum_{i, j}\left(-\overline{[x]_i}[W]_{ij}[y]_j+\overline{[y]_i}[W]_{ij}[x]_j\right) =\sum_{i, j}[W]_{ij}(\overline{[y]_i}[x]_j-\overline{[x]_i}[y]_j).
\end{aligned}
\end{align*}
This and the Cauchy-Schwarz inequality yield
\begin{align*}
\begin{aligned}
&|\langle Wx, y\rangle+\langle y, Wx\rangle|^2 \\
&\hspace{0.5cm} =\left(\sum_{i, j}[W]_{ij}(\overline{[y]_i}[x]_j-\overline{[x]_i}[y]_j)\right)^2 \leq \left(\sum_{i, j}|[W]_{ij}|^2\right)\left(\sum_{i, j} \Big|\overline{[y]_i}[x]_j-\overline{[x]_i}[y]_j \Big|^2\right)\\
&\hspace{0.5cm} = \|W\|_F^2 \Big(2\|x\|^2\|y\|^2-\langle x, y\rangle^2-\langle y, x\rangle^2 \Big) =2\|W\|_F^2\cdot\left(\|x\|^2\|y\|^2-\mathrm{Re}(\langle x, y\rangle^2)\right).
\end{aligned}
\end{align*}
This implies the desired estimate.
\end{proof}
\begin{remark} \label{R4.1} Let $z_i$ and $z_j$ be complex vectors with $\|z_i \|= \|z_j \| = 1$. Then, we use Lemma \ref{L4.1} to get 
\[
\|\langle Wz_i, z_j\rangle+\langle z_j, Wz_i\rangle|\leq \sqrt{2}\|W\|_F\cdot\sqrt{1-\mathrm{Re}(\langle z_i, z_j\rangle^2)}.
\]
\end{remark}
In the following subsections, we study emergent dynamics of \eqref{LHSF} with $\kappa_1 = 0$ and the full system separately. \newline

Consider the system which can be obtained from the full system \eqref{LHSF} with $\kappa_1 = 0$:
\begin{equation} \label{D-1}
\begin{cases} 
\displaystyle \dot{z}_j=\Omega_j z_j+\kappa_0(\langle{z_j, z_j}\rangle V_0 z_c-\langle{V_0 z_c, z_j}\rangle z_j), \\
\displaystyle V_0=I_{d+1}+W_0.
\end{cases}
\end{equation}
or equivalently
\begin{equation*} \label{D-2}
 \dot{z}_j=\Omega_j z_j+ \kappa_0 \Big(\langle{z_j, z_j}\rangle  z_c-\langle{z_c, z_j}\rangle z_j \Big) +  \kappa_0 \Big(\langle{z_j, z_j}\rangle W z_c-\langle{W z_c, z_j}\rangle z_j \Big).
\end{equation*}
Next we consider the temporal evolution of $\langle z_i, z_j \rangle$. 
\begin{lemma}\label{L4.2}
Let $\{z_j\}$ be the solution of the system \eqref{D-1}. Then, one has
\[
\frac{d}{dt}\langle{z_i, z_j}\rangle=\langle (\Omega_i-\Omega_j)z_i, z_j\rangle+\kappa_0 \Big (1-\langle{z_i, z_j}\rangle)(\langle{V_0z_c, z_j}\rangle+\langle{z_i, V_0z_c}\rangle \Big).
\]
\end{lemma}
\begin{proof}
Note that 
\begin{equation}  \label{D-3}
\frac{d}{dt}\langle{z_i, z_j}\rangle =\langle{\dot{z}_i, z_j}\rangle+\langle{z_i, \dot{z}_j}\rangle.
\end{equation}
For each term in the R.H.S. of \eqref{D-3}, we use \eqref{D-1} to see
\begin{align}
\begin{aligned} \label{D-4}
\langle{\dot{z}_i, z_j}\rangle&=\langle \Omega_i z_i, z_j\rangle+\kappa_0(\langle{V_0z_c, z_j}\rangle-\langle{z_i, V_0z_c}\rangle\langle{z_i, z_j}\rangle),\\
\langle{z_i, \dot{z}_j}\rangle &=\langle z_i, \Omega_j z_j\rangle+\kappa_0(\langle{z_i, V_0z_c}\rangle-\langle{V_0z_c, z_j}\rangle\langle{z_i, z_j}\rangle),
\end{aligned}
\end{align}
Now, we combine \eqref{D-3} and \eqref{D-4} to get 
\begin{align*}
\frac{d}{dt}\langle{z_i, z_j}\rangle&=\langle{\dot{z}_i, z_j}\rangle+\langle{z_i, \dot{z}_j}\rangle\\
&=\langle \Omega_i z_i, z_j\rangle+\langle z_i, \Omega_jz_j\rangle \\
&\hspace{0.2cm} +\kappa_0 \Big(\langle{V_0z_c, z_j}\rangle-\langle{z_i, V_0z_c}\rangle\langle{z_i, z_j}\rangle+\langle{z_i, V_0z_c}\rangle-\langle{V_0z_c, z_j}\rangle\langle{z_i, z_j}\rangle \Big)\\
&=\langle (\Omega_i-\Omega_j)z_i, z_j\rangle+\kappa_0(1-\langle{z_i, z_j}\rangle)(\langle{V_0z_c, z_j}\rangle+\langle{z_i, V_0z_c}\rangle).
\end{align*}
\end{proof}

\subsection{Identical ensemble} \label{sec:4.1}
In this subsection, we study emergent behaviors of identical ensemble. For four points $(z_i, z_j, z_k, z_l)$ lying in a general position, we introduce a cross-ratio-like quantity \cite{H-K-P-R, Lo-0}:
\[
\mathcal{C}_{ijkl}[{\mathcal Z}]:=\frac{(1-\langle{z_i, z_j}\rangle)(1-\langle{z_k, z_l}\rangle)}{(1-\langle{z_i, z_l}\rangle)(1-\langle{z_k, z_j}\rangle)}.
\] 
Next, we show that $\mathcal{C}_{ijkl}$ is conserved along the dynamics \eqref{D-1} with $\Omega_i = \Omega$. 
\begin{proposition} \label{P4.1}
Suppose the natural frequency matrices are the same:
\[ \Omega_j = \Omega, \quad j = 1, \cdots, N, \]
and let $\{z_j\}$ be a global solution of system \eqref{D-1}. Then, one has
\[ \mathcal{C}_{ijkl}[{\mathcal Z}(t)] = \mathcal{C}_{ijkl}[{\mathcal Z}(0)], \quad t \geq 0. \]
\end{proposition}
\begin{proof}
Since $\Omega_i=\Omega_j$, we have
\begin{equation} \label{D-5}
\frac{d}{dt}\langle z_i, z_j\rangle=\kappa_0(1-\langle z_i, z_j\rangle) \Big(\langle V_0 z_c, z_j\rangle+\langle z_i, V_0z_c\rangle \Big).
\end{equation}
This implies
\begin{equation} \label{D-6}
\frac{d}{dt}(1-\langle z_i, z_j\rangle)=-\kappa_0(1-\langle z_i, z_j\rangle)(\langle V_0 z_c, z_j\rangle+\langle z_i, V_0z_c\rangle).
\end{equation}
Again this yields the desired estimate:
\begin{align*}
\frac{d}{dt}\mathcal{C}_{ijkl} &=\frac{d}{dt}\left(\frac{(1-\langle{z_i, z_j}\rangle)(1-\langle{z_k, z_l}\rangle)}{(1-\langle{z_i, z_l}\rangle)(1-\langle{z_k, z_j}\rangle)}\right) \\
&=-\kappa_0\left(\frac{(1-\langle{z_i, z_j}\rangle)(1-\langle{z_k, z_l}\rangle)}{(1-\langle{z_i, z_l}\rangle)(1-\langle{z_k, z_j}\rangle)}\right) \\
&\hspace{0.5cm}\times \Big (\langle V_0z_c, z_j\rangle+\langle z_i,V_0z_c\rangle+\langle V_0z_c, z_l\rangle+\langle z_k,V_0z_c\rangle -\langle V_0z_c, z_l\rangle \\
&\hspace{2cm}-\langle z_i,V_0z_c\rangle-\langle V_0z_c, z_j\rangle-\langle z_k,V_0z_c\rangle \Big)=0.
\end{align*}
\end{proof}

\vspace{0.5cm}

Next we use \eqref{D-6} to obtain
\begin{align}
\begin{aligned}\label{D-6-1}
& \frac{d}{dt}|1-\langle{z_i, z_j}\rangle|^2 \\
& \hspace{0.5cm} =\frac{d}{dt}\big((1-\langle{z_i, z_j}\rangle)(1-\langle{z_j, z_i}\rangle)\big)\\
& \hspace{0.5cm} =-\kappa_0(1-\langle{z_i, z_j}\rangle)(\langle{V_0z_c, z_j}\rangle+\langle{z_i, V_0z_c}\rangle)(1-\langle{z_j, z_i}\rangle)+(c.c.)\\
& \hspace{0.5cm} =-\kappa_0|1-\langle{z_i, z_j}\rangle|^2(\langle{V_0z_c, z_j}\rangle+\langle{z_i, V_0z_c}\rangle+\langle{V_0z_c, z_i}\rangle+\langle{z_j, V_0z_c}\rangle)\\
& \hspace{0.5cm} =-\kappa_0|1-\langle{z_i, z_j}\rangle|^2 \Big(\langle{V_0z_c, z_i+z_j}\rangle+\langle{z_i+z_j, V_0z_c}\rangle \Big ).
\end{aligned}
\end{align}
Now we introduce new variables defined as follows:
\begin{equation} \label{D-7}
R_{ij} :=\mathrm{Re}(\langle z_i, z_j\rangle),\quad I_{ij} :=\mathrm{Im}(\langle z_i, z_j\rangle)\quad\forall~i, j\in\{1, 2, \cdots, N\}.
\end{equation}
Since $\langle z_i, z_i \rangle = 1$, it is easy to see that 
\[  R_{ii}=1 \quad \mbox{and} \quad I_{ii}=0, \quad i = 1, \cdots, N. \]
\begin{lemma} \label{L4.3}
Let $R_{ij}$ and $I_{ij}$ be quantities defined by the relation \eqref{D-7}. Then, the following assertions hold.  
\begin{enumerate}
\item 
$R_{ij}$ and $I_{ij}$ satisfy symmetry-antisymmetry relations:
\[ R_{ij}=R_{ji} \quad \mbox{and} \quad I_{ij}=-I_{ij}, \quad 1 \leq i, j  \leq N. \]
\item
$(R_{ij}, I_{ij})$ lies on the unit ball:
\[  |R_{ij}|^2+ |I_{ij}|^2\leq 1, \quad 1 \leq i, j \leq N. \]
\end{enumerate}
\end{lemma}
\begin{proof}(i)~Note that 
\[ \langle z_i, z_j \rangle = \overline{\langle z_j, z_i \rangle}, \quad \mbox{i.e.,} \quad R_{ij} + {\mathrm i} I_{ij} = \overline{R_{ji} + {\mathrm i} I_{ji}} = R_{ji} - {\mathrm i} I_{ji}. \]
We compare the real and imaginary parts of the above relation to find the proof of the first assertion. \newline

\noindent (ii)~The second assertion follows from the relation $|\langle z_i, z_j \rangle| \leq 1$. 
\end{proof}

\vspace{0.2cm}

Now, we set  
\[
{\mathcal J}_{ij} :=\sqrt[4]{(1-R_{ij})^2+I_{ij}^2}, \quad \forall i, j\in\{1, 2, \cdots, N\}.
\]
Then, note that the zero convergence of  ${\mathcal J}_{ij}$ yields the emergence of the complete aggregation.

\begin{lemma} \label{L4.4}
Suppose the natural frequency matrices are the same:
\[ \Omega_j = \Omega, \quad j = 1, \cdots, N, \]
and let $\{z_j\}$ be a global solution of system \eqref{D-1}.  Then, the functional ${\mathcal J}_{ij}$ satisfies
\[
\frac{d\mathcal J_{ij}}{dt} \leq -\frac{\kappa_0}{2N} \mathcal J_{ij}\sum_{k=1}^N\left(1 -\mathcal J_{ik}^2 -\frac{\sqrt[4]{5}}{\sqrt{2}} \|W_0\|_F \mathcal J_{ik} +1 -\mathcal J_{jk}^2 -\frac{\sqrt[4]{5}}{\sqrt{2}} \|W_0\|_F \mathcal J_{jk}\right).
\]
\end{lemma}
\begin{proof}
From \eqref{D-6-1}, we have
\begin{align*}
& \frac{d}{dt}((1-R_{ij})^2+I_{ij}^2) \\
& \hspace{0.5cm} =-\frac{\kappa_0}{N}((1-R_{ij})^2+I_{ij}^2)\sum_{k=1}^N\left(2R_{ik}+2R_{jk}+\langle W_0 z_k, z_i+z_j\rangle+\langle z_i+z_j, W_0z_k\rangle \right).
\end{align*}
Also, we can apply Lemma \ref{L4.1} to obtain
\[
|\langle W_0 z_k, z_i\rangle+\langle z_i, W_0z_k\rangle|\leq \sqrt{2}\|W_0\|_F\sqrt{1-R_{ik}^2+I_{ik}^2}.
\]
Moreover, from the relation $R_{ik}^2+I_{ik}^2\leq 1$ and Cauchy-Schwarz inequality, we have
\begin{align*}
1-R_{ik}^2+I_{ik}^2&=(1+R_{ik})(1-R_{ik})+I_{ik}^2\\
&\leq2(1-R_{ik})+I_{ik}\leq\sqrt{(4+1)\Big((1-R_{ik})^2+I_{ik}^2\Big)}\\
&=\sqrt{5\Big((1-R_{ik})^2+I_{ik}^2\Big)}.
\end{align*}
If we combine above calculations, we have
\begin{align*}
\begin{aligned}
&\frac{d}{dt}((1-R_{ij})^2+I_{ij}^2)  \leq -\frac{\kappa_0}{N}((1-R_{ij})^2+I_{ij}^2) \sum_{k=1}^N\bigg( 2R_{ik}+2R_{jk} \\
& \hspace{1cm} -\sqrt{2}\|W_0\|_F\sqrt[4]{5 \big[ (1-R_{ik})^2+I_{ik}^2 \big]} -\sqrt{2}\|W_0\|_F\sqrt[4]{5\big[ (1-R_{jk})^2+I_{jk}^2 \big]} \bigg).
\end{aligned}
\end{align*}
Then, we can rewrite above relation as
\[
\frac{d\mathcal{J}_{ij}^4}{dt} \leq -\frac{\kappa_0}{N}\mathcal{J}_{ij}^4\cdot\sum_{k=1}^N\left(2R_{ik}+2R_{jk}-\sqrt{2}\sqrt[4]{5}\|W_0\|_F(\mathcal{J}_{ik}+\mathcal{J}_{jk})\right)
\]
and this implies
\[
\frac{d\mathcal J_{ij}}{dt} \leq -\frac{\kappa_0}{2N}\mathcal J_{ij}\sum_{k=1}^N\left(R_{ik}+R_{jk} -\frac{\sqrt[4]{5}}{\sqrt{2}} \|W_0\|_F \mathcal J_{ik} -\frac{\sqrt[4]{5}}{\sqrt{2}} \|W_0\|_F \mathcal J_{jk}\right).
\]
Note that the definition of $\mathcal J_{ij}$ yeilds
\[
R_{ij}=1-\sqrt{\mathcal J_{ij}^4 -I_{ij}^2}\geq 1-\mathcal J_{ij}^2,
\]
from which we can get the desired result:
\begin{align}\label{C-1-1}
\frac{d\mathcal J_{ij}}{dt} \leq -\frac{\kappa_0}{2N} \mathcal J_{ij}\sum_{k=1}^N\left(1 -\mathcal J_{ik}^2 -\frac{\sqrt[4]{5}}{\sqrt{2}} \|W_0\|_F \mathcal J_{ik} +1 -\mathcal J_{jk}^2 -\frac{\sqrt[4]{5}}{\sqrt{2}} \|W_0\|_F \mathcal J_{jk}\right).
\end{align}
\end{proof}
Now, we are ready to present our third main result. 
\begin{theorem}\label{T4.1}
Suppose the natural frequency matrix and initial data satisfy 
\begin{equation} \label{C-1-2}
 \Omega_j = \Omega, \quad {\mathcal J}_{ij}^{in} < \frac{2\sqrt{2}}{\sqrt{\sqrt{5} \|W_0\|_F^2 +8}+\sqrt[4]{5} \|W_0\|_F}, \quad i, j = 1, \cdots, N, 
\end{equation} 
and let $\{z_j\}$ be the solution for \eqref{D-1} with initial data $\{z_j^{in}\}$. Then, there exists a positive constant $\Lambda_{ij}$ such that 
\[
|\mathcal J_{ij}(t)|\lesssim\exp\left(-\alpha_{ij}t\right),\quad \|z_i(t)-z_j(t)\|\lesssim\exp\left(-\frac{\Lambda_{ij}t}{2}\right),
\]
where
\[
\Lambda_{ij} := \frac{\kappa_0}{2N} \sum_{k=1}^N\left(1 -(\mathcal J_{ik}^{in})^2 -\frac{\sqrt[4]{5}}{\sqrt{2}} \|W_0\|_F \mathcal J_{ik}^{in} +1 -(\mathcal J_{jk}^{in})^2 -\frac{\sqrt[4]{5}}{\sqrt{2}} \|W_0\|_F \mathcal J_{jk}^{in} \right).
\]
\end{theorem}

\begin{proof}
From the assumption \eqref{C-1-2} and the inequality \eqref{C-1-1}, we know that $\mathcal J_{ij}$ always decreases. This implies
\[
\mathcal J_{ij}(t)\leq \mathcal J_{ij}^{in}, \quad t\geq0,\quad i, j\in\{1, 2, \cdots, N\}.
\]
Then we can also obtain
\begin{align*}
\frac{d\mathcal J_{ij}}{dt} &\leq -\frac{\kappa_0}{2N} \mathcal J_{ij}\sum_{k=1}^N\left(1 -\mathcal J_{ik}^2 -\frac{\sqrt[4]{5}}{\sqrt{2}} \|W_0\|_F \mathcal J_{ik} +1 -\mathcal J_{jk}^2 -\frac{\sqrt[4]{5}}{\sqrt{2}} \|W_0\|_F \mathcal J_{jk}\right) \\
&\leq -\frac{\kappa_0}{2N} \mathcal J_{ij} \sum_{k=1}^N\left(1 -(\mathcal J_{ik}^{in})^2 -\frac{\sqrt[4]{5}}{\sqrt{2}} \|W_0\|_F \mathcal J_{ik}^{in} +1 -(\mathcal J_{jk}^{in})^2 -\frac{\sqrt[4]{5}}{\sqrt{2}} \|W_0\|_F \mathcal J_{jk}^{in} \right).
\end{align*}
From the assumption \eqref{C-1-2}, we have 
\[ \Lambda_{ij}>0 \quad \mbox{for all $i, j\in\{1, 2, \cdots, N\}$}. \] 
Then, we have
\[
\frac{d\mathcal J_{ij}}{dt} \leq -\Lambda_{ij} \mathcal J_{ij}.
\]
This implies
\begin{align*}
\mathcal J_{ij}(t)\leq \mathcal J_{ij}^{in}\exp(-\Lambda_{ij}t) \quad \mbox{and} \quad \|z_i-z_j\|^2=2-2R_{ij}\leq 2 \mathcal J_{ij}\leq2 \mathcal J_{ij}^{in}\exp(-\Lambda_{ij}t).
\end{align*}
\end{proof}

\subsection{Non-identical ensemble} \label{sec:4.2}
In this subsection, we study practical aggregation of non-identical ensemble to \eqref{D-1}. 

\begin{lemma} \label{L4.5}
Let $\{z_j\}$ be a global solution of system \eqref{D-1}.  Then, the functional ${\mathcal J}_{ij}$ satisfies
\[
\frac{d\mathcal J_{ij}}{dt} \leq\frac{3\sqrt{2}}{4\mathcal J_{ij}}\mathcal{D}(\Omega) -\frac{\kappa_0}{2N} \mathcal J_{ij}\sum_{k=1}^N\left(1 -\mathcal J_{ik}^2 -\frac{\sqrt[4]{5}}{\sqrt{2}} \|W_0\|_F \mathcal J_{ik} +1 -\mathcal J_{jk}^2 -\frac{\sqrt[4]{5}}{\sqrt{2}} \|W_0\|_F \mathcal J_{jk}\right).
\]
\end{lemma}
\begin{proof}
As in \eqref{D-6-1}, one has 
\begin{align*}
\frac{d}{dt}|1-\langle z_i, z_j\rangle|^2&=(1-\langle z_j, z_i\rangle)\langle (\Omega_i-\Omega_j)z_i, z_j\rangle+(1-\langle z_i, z_j\rangle)\langle (\Omega_j-\Omega_i)z_j, z_i\rangle\\
&-\kappa_0|1-\langle{z_i, z_j}\rangle|^2(\langle{V_0z_c, z_i+z_j}\rangle+\langle{z_i+z_j, V_0z_c}\rangle)
\end{align*}
Now, we use Lemma \ref{L4.1} to get 
\begin{align*}
\begin{aligned}
&\left|(1-\langle z_j, z_i\rangle)\langle (\Omega_i-\Omega_j)z_i, z_j\rangle+(1-\langle z_i, z_j\rangle)\langle (\Omega_j-\Omega_i)z_j, z_i\rangle\right|\\
&\hspace{.5cm} =\left|\langle (\Omega_i-\Omega_j)z_i, (1-\langle z_j, z_i\rangle)z_j\rangle+\langle (1-\langle z_j, z_i\rangle)z_j, (\Omega_i-\Omega_j)z_i\rangle\right|\\
&\hspace{.5cm}  \leq \sqrt{2}\|\Omega_i-\Omega_j\|_F\sqrt{\|z_i\|^2\|(1-\langle z_j, z_i\rangle)z_j\|^2-\mathrm{Re} \big[ \langle z_i, (1-\langle z_j, z_i\rangle)z_j \rangle^2 \big]} \\
&\hspace{.5cm}  =\sqrt{2}\|\Omega_i-\Omega_j\|_F\sqrt{|1-\langle z_j, z_i\rangle|^2-\mathrm{Re} \big[ \langle z_i, z_j\rangle^2(1-\langle z_j, z_i\rangle)^2 \big]}\\
&\hspace{.5cm}  =\sqrt{2}\|\Omega_i-\Omega_j\|_F\sqrt{(1-R_{ij})^2+I_{ij}^2-\mathrm{Re} \big[ (R_{ij}+\mathrm{i}I_{ij}-R_{ij}^2-I_{ij}^2)^2 \big]} \\
&\hspace{.5cm}  =\sqrt{2}\|\Omega_i-\Omega_j\|_F\sqrt{(1-R_{ij})^2+I_{ij}^2-(R_{ij}-R_{ij}^2-I_{ij}^2)^2+I_{ij}^2}\\
&\hspace{.5cm}  =\sqrt{2}\|\Omega_i-\Omega_j\|_F\sqrt{(1-R_{ij}^2-I_{ij}^2)((1-R_{ij})^2+I_{ij}^2)+2I_{ij}^2}\\
&\hspace{.5cm} \leq \sqrt{2}\|\Omega_i-\Omega_j\|_F\sqrt{ \mathcal J_{ij}^4(3-R_{ij}^2-I_{ij}^2)}\leq 3\sqrt{2}\mathcal{D}(\Omega) \mathcal J_{ij}^2,
\end{aligned}
\end{align*}
where $\mathcal{D}(\Omega)=\max_{i, j}\|\Omega_i-\Omega_j\|_F$. Finally we can obtain the desired estimate.
\begin{align}\label{C-2}
\frac{d\mathcal J_{ij}}{dt} \leq \frac{3\sqrt{2}}{4\mathcal J_{ij}}\mathcal{D}(\Omega) -\frac{\kappa_0}{2N} \mathcal J_{ij}\sum_{k=1}^N\left(1 -\mathcal J_{ik}^2 -\frac{\sqrt[4]{5}}{\sqrt{2}} \|W_0\|_F \mathcal J_{ik} +1 -\mathcal J_{jk}^2 -\frac{\sqrt[4]{5}}{\sqrt{2}} \|W_0\|_F \mathcal J_{jk}\right).
\end{align}
\end{proof}
\begin{theorem}\label{T4.2}
Suppose initial data $\{z_j^{in}\}$ satisfy \eqref{C-1-2} and let $\{z_j\}$ be the solution of \eqref{D-1}. Then, we have following practical aggregation:
\[
\lim_{\kappa_0\to\infty}\limsup_{t\to\infty} \mathcal J_{ij}=0.
\]
\end{theorem}
\begin{proof}
We will use a similar argument with Theorem \ref{T3.1}. First, consider the following quartic polynomial: 
\[
p(J)=\kappa_0 J^2 \left( J^2 +\frac{\sqrt[4]{5}}{\sqrt{2}} \|W_0\|_F J -1 \right) +\frac{3\sqrt{2}}{4}\mathcal{D}(\Omega).
\]
For a sufficiently large $\kappa_0$, the polynomial $p$ has two positive solutions, called $J_\pm$, which satisfy
\[
\lim_{\kappa_0\to\infty}J_{+} = \frac{2\sqrt{2}}{\sqrt{\sqrt{5} \|W_0\|_F^2 +8}-\sqrt[4]{5} \|W_0\|_F}, \quad \lim_{\kappa_0\to\infty}J_-=0.
\]
If we set
\[
\mathcal J_M(t)=\max_{i, j} \mathcal J_{ij}(t),
\]
we have
\[
\mathcal J_M\frac{d\mathcal J_M}{dt} \leq \frac{3\sqrt{2}}{4}\mathcal{D}(\Omega) +\kappa_0 \mathcal J_M^2 \left(\mathcal J_M^2 +\frac{\sqrt[4]{5}}{\sqrt{2}} \|W_0\|_F \mathcal J_M -1 \right).
\]
Hence, we can obtain that if
\[
\mathcal J_M(0) < \frac{2\sqrt{2}}{\sqrt{\sqrt{5} \|W_0\|_F^2 +8}+\sqrt[4]{5} \|W_0\|_F},
\]
then we can obtain following practical aggregation:
\[
\lim_{\kappa_0\to\infty}\lim_{t\to\infty}\mathcal J_M=0.
\]
\end{proof}

\begin{remark}\label{R4.2}
Since the leading coefficient of $p$ is $\mathcal O(\kappa_0)$ and $p$ is quartic, for sufficiently large $\kappa_0$, one has $$J_- =\mathcal O\Big( \kappa_0^{-\frac12} \Big).$$
\end{remark}

\section{Emergent dynamics of the full LHS model} \label{sec:5}
\setcounter{equation}{0}
In this section, we study emergent dynamics of the LHS model. In other words, we need to consider the effect of
\begin{align*}
\kappa_1(\langle{z_j, V_1 z_c}\rangle-\langle{V_1 z_c, z_j}\rangle)z_j,\quad V_1=I_{d+1}+W_1.
\end{align*}

\begin{lemma} \label{L5.1}
Let $\{z_j\}$ be a global solution of system \eqref{LHSF}. Then, the functional ${\mathcal J}_{ij}$ satisfies
\begin{align*}
\frac{d\mathcal J_{ij}}{dt} & \leq\frac{3\sqrt{2}}{4\mathcal J_{ij}}\mathcal{D}(\Omega) -\frac{\kappa_0}{2N} \mathcal J_{ij}\sum_{k=1}^N\left(1 -\mathcal J_{ik}^2 -\frac{\sqrt[4]{5}}{\sqrt{2}} \|W_0\|_F \mathcal J_{ik} +1 -\mathcal J_{jk}^2 -\frac{\sqrt[4]{5}}{\sqrt{2}} \|W_0\|_F \mathcal J_{jk}\right) \\
& \hspace{.5cm} +\frac{\kappa_1}{N \mathcal J_{ij}}\sum_{k=1}^N(\mathcal J_{ik}^2 +\mathcal J_{kj}^2) +\sqrt{2}\kappa_1 \|W_1\|_F.
\end{align*}
\end{lemma}
\begin{proof}
It suffices to consider additional term including $\kappa_1$ and then, we can use the calculation in proofs of Lemma \ref{L4.2} and \ref{L4.4}:
\begin{align*}
& \kappa_1(\langle z_j, V_1z_c\rangle-\langle V_1z_c,z_j\rangle)\langle z_i, z_j\rangle+\kappa_1(\langle V_1z_c, z_i\rangle-\langle z_i, V_1z_c\rangle)\langle z_i, z_j\rangle \\
&\hspace{1cm} = \kappa_1\langle z_i, z_j\rangle(\langle z_j-z_i, V_1z_c\rangle-\langle V_1z_c, z_j-z_i\rangle).
\end{align*}
From this calculation, we have
\begin{align*}
& \kappa_1\langle z_i, z_j\rangle(1-\langle z_j, z_i\rangle)(\langle z_j-z_i, V_1z_c\rangle-\langle V_1z_c, z_j-z_i\rangle) \\
&\hspace{1cm}+\kappa_1\langle z_j, z_i\rangle(1-\langle z_i, z_j\rangle)(\langle z_i-z_j, V_1z_c\rangle-\langle V_1z_c, z_i-z_j\rangle) \\
& \hspace{1cm} =\kappa_1(\langle z_i, z_j\rangle-\langle z_j, z_i\rangle)(\langle z_j-z_i, V_1z_c\rangle-\langle V_1z_c, z_j-z_i\rangle)\\
& \hspace{1cm} =-\frac{2\kappa_1\mathrm{i}I_{ij}}{N}\sum_{k=1}^N(\langle z_i-z_j, V_1z_k\rangle-\langle V_1z_k, z_i-z_j\rangle).
\end{align*}
We also have
\begin{align*}
&\langle z_i-z_j, V_1z_k\rangle-\langle V_1z_k, z_i-z_j\rangle\\
& \hspace{0.5cm} =\langle z_i, z_k\rangle-\langle z_j, z_k\rangle-\langle z_k, z_i\rangle+\langle z_k, z_j\rangle+\langle z_i-z_j, W_1 z_k\rangle-\langle W_1 z_k, z_i-z_j\rangle\\
& \hspace{0.5cm} =2\mathrm{i}I_{ik}+2\mathrm{i}I_{kj}+\langle z_i-z_j, W_1 z_k\rangle-\langle W_1 z_k, z_i-z_j\rangle.
\end{align*}
We combine this identity with the previous results in proofs of Lemma \ref{L4.2} and \ref{L4.4} to have
\begin{align}
\begin{aligned}\label{C-3}
& \frac{d}{dt}|1-\langle z_i, z_j\rangle|^2 \\
& \hspace{0.2cm} = (1-\langle z_j, z_i\rangle)\langle (\Omega_i-\Omega_j)z_i, z_j\rangle+(1-\langle z_i, z_j\rangle)\langle (\Omega_j-\Omega_i)z_j, z_i\rangle \\
&\hspace{0.5cm} -\frac{\kappa_0}{N}((1-R_{ij})^2+I_{ij}^2)\sum_{k=1}^N\Big(2R_{ik}+2R_{jk}+\langle W_0 z_k, z_i+z_j\rangle+\langle z_i+z_j, W_0z_k\rangle\Big) \\
&\hspace{0.5cm} -\frac{4\kappa_1 I_{ij}}{N}\sum_{k=1}^N(I_{ik}+I_{kj})+\frac{2\kappa_1\mathrm{i}I_{ij}}{N}\sum_{k=1}^N\Big(\langle z_i-z_j, W_1z_k\rangle-\langle W_1z_k, z_i-z_j\rangle\Big).
\end{aligned}
\end{align}

\noindent Finally, we combine the relations \eqref{C-2} and \eqref{C-3} to obtain following result:
\begin{align*}
\frac{d \mathcal J_{ij}}{dt} & \leq\frac{3\sqrt{2}}{4\mathcal J_{ij}}\mathcal{D}(\Omega) -\frac{\kappa_0}{2N} \mathcal J_{ij}\sum_{k=1}^N\left(1 -\mathcal J_{ik}^2 -\frac{\sqrt[4]{5}}{\sqrt{2}} \|W_0\|_F \mathcal J_{ik} +1 -\mathcal J_{jk}^2 -\frac{\sqrt[4]{5}}{\sqrt{2}} \|W_0\|_F \mathcal J_{jk}\right) \\
& \hspace{.5cm} -\left| \frac{\kappa_1 I_{ij}}{N \mathcal J_{ij}^3}\sum_{k=1}^N(I_{ik}+I_{kj}) -\frac{\kappa_1\mathrm{i}I_{ij}}{2N \mathcal J_{ij}^3}\sum_{k=1}^N\Big(\langle z_i-z_j, W_1z_k\rangle-\langle W_1z_k, z_i-z_j\rangle\Big) \right|.
\end{align*}
From the definition of $\mathcal J_{ij}$, we have
\[
|I_{ij}|\leq \mathcal J_{ij}^2, \quad |I_{ik}|\leq \mathcal J_{ik}^2, \quad |I_{kj}|\leq \mathcal J_{kj}^2,
\]
and also we have
\[
|\langle z_i-z_j, W_1z_k\rangle|\leq \|W_1\|_F\cdot\|z_i-z_j\|=\|W_1\|_F\sqrt{2(1-R_{ij})}\leq \sqrt{2}\|W_1\|_F \mathcal J_{ij}.
\]
Finally, we can obtain the desired estimate:
\begin{align*}
\frac{d\mathcal J_{ij} }{dt} & \leq\frac{3\sqrt{2}}{4\mathcal J_{ij}}\mathcal{D}(\Omega) -\frac{\kappa_0}{2N} \mathcal J_{ij}\sum_{k=1}^N\left(1 -\mathcal J_{ik}^2 -\frac{\sqrt[4]{5}}{\sqrt{2}} \|W_0\|_F \mathcal J_{ik} +1 -\mathcal J_{jk}^2 -\frac{\sqrt[4]{5}}{\sqrt{2}} \|W_0\|_F \mathcal J_{jk}\right) \\
& \hspace{.5cm} +\frac{\kappa_1}{N \mathcal J_{ij}}\sum_{k=1}^N(\mathcal J_{ik}^2 +\mathcal J_{kj}^2) +\sqrt{2}\kappa_1 \|W_1\|_F.
\end{align*}
\end{proof}
For a configuration $z$, we set 
\[ \mathcal J_M(t) := \max_{i, j} \mathcal J_{ij}. \]
\begin{theorem} \label{T5.1}
Suppose system parameters and initial data satisfy
\begin{align}
\begin{aligned} \label{C-1-4}
& \kappa_0 > 2\kappa_1\geq0, \quad W_1 \equiv 0, \quad \mathcal{D}(\Omega)=0, \\
& {\mathcal J}^{in}_{ij} < \frac{2\sqrt{2}\left( 1-\frac{2\kappa_1}{\kappa_0} \right)}{\sqrt{\sqrt{5} \|W_0\|_F^2 +8\left( 1-\frac{2\kappa_1}{\kappa_0} \right)}+\sqrt[4]{5} \|W_0\|_F} \quad \forall ~i,j = 1,\cdots, N,
\end{aligned}
\end{align}
and let $\{z_j\}$ be the solution of \eqref{LHSF} with initial data $\{z_j^{in}\}$. Then, there exists a positive constant ${\tilde \Lambda}$ such that 
\[
\mathcal J_M(t)\leq \mathcal J_M(0)\exp\left(-{\tilde \Lambda} t\right), \quad t > 0,
\]
where
\[ {\tilde \Lambda} := \kappa_0\left(1-\frac{2\kappa_1}{\kappa_0} -|\mathcal J_M^{in}|^2 -\frac{\sqrt[4]{5}}{\sqrt{2}} \|W_0\|_F \mathcal J_M^{in} \right).
\]
\end{theorem}
\begin{proof}
Again, we use Lemma \ref{L5.1} to obtain
\begin{align}
\begin{aligned}\label{C-4}
\frac{d\mathcal J_M}{dt} & \leq\frac{3\sqrt{2}}{4\mathcal J_M}\mathcal{D}(\Omega) -\kappa_0 \mathcal J_M \left(1 -\mathcal J_M^2 -\frac{\sqrt[4]{5}}{\sqrt{2}} \|W_0\|_F \mathcal J_M \right) +2\kappa_1\mathcal J_M +\sqrt{2}\kappa_1 \|W_1\|_F.
\end{aligned}
\end{align}

\noindent By the assumption $\|W_1\|_F=0$ and $\mathcal{D}(\Omega)=0$, one has
\begin{align*}
\frac{d\mathcal J_M}{dt} & \leq -\kappa_0 \mathcal J_M \left(1-\frac{2\kappa_1}{\kappa_0} -\mathcal J_M^2 -\frac{\sqrt[4]{5}}{\sqrt{2}} \|W_0\|_F \mathcal J_M \right).
\end{align*}
If we use the same argument with Theorem \ref{T4.1} and assumption on the initial data, we obtain
\[
J_M(t)\leq J_M(0)\exp\left(-{\tilde \Lambda} t\right).
\]
\end{proof}

\begin{theorem}\label{T5.2}
Let $\{z_j\}$ be the solution of \eqref{LHSF} with the initial data $\{z_j^{in}\}=$ satisfying \eqref{C-1-2}. Then, for fixed $\kappa_1$, we can obtain following practical aggregation
\[
\lim_{\kappa_0\to\infty}\limsup_{t\to\infty}J_M(t)=0.
\]
\end{theorem}
\begin{proof} We first rewrite \eqref{C-4} as
\begin{align*}
\frac{d\mathcal J_M}{dt} & \leq \frac{3\sqrt{2}}{4\mathcal J_M}\mathcal{D}(\Omega) +\sqrt{2}\kappa_1 \|W_1\|_F -\kappa_0 \mathcal J_M \left(1-\frac{2\kappa_1}{\kappa_0} -\mathcal J_M^2 -\frac{\sqrt[4]{5}}{\sqrt{2}} \|W_0\|_F \mathcal J_M \right).
\end{align*}
From the similar argument with Theorem \ref{T4.2}, we can also define quartic polynomial as follows:
\begin{align*}
p(J)&= \frac{3\sqrt{2}}{4}\mathcal{D}(\Omega) +\sqrt{2}\kappa_1 \|W_1\|_F J -\kappa_0 J^2 \left(1-\frac{2\kappa_1}{\kappa_0} -J^2 -\frac{\sqrt[4]{5}}{\sqrt{2}} \|W_0\|_F J \right) \\
&=\kappa_0 J^2 \left(J^2 +\frac{\sqrt[4]{5}}{\sqrt{2}} \|W_0\|_F J -1 \right) +\kappa_1 J \left( 2J +\sqrt{2} \|W_1\|_F \right) +\frac{3\sqrt{2}}{4}\mathcal{D}(\Omega).
\end{align*}
Then, for sufficiently large $\kappa_0$ relative to $\kappa_1$, we can also set two positive solutions of $p$ as $J_\pm$. Also we know that
\[
\lim_{\kappa_0\to\infty}J_+=\frac{2\sqrt{2}}{\sqrt{\sqrt{5} \|W_0\|_F^2 +8} +\sqrt[4]{5} \|W_0\|_F},\quad \lim_{\kappa_0\to\infty}J_-=0.
\]
Hence, assumption on initial condition implies the desired result.
\end{proof}

\begin{remark}\label{R5.1}
As in Remark \ref{R4.2}, for sufficiently large $\kappa_0$, one has $$J_- =\mathcal O \Big( \kappa_0^{-\frac12} \Big).$$
\end{remark}

\section{Numerical simulation} \label{sec:6}
\setcounter{equation}{0}
In this section, we provide several numerical examples in order to confirm analytical results in Sections \ref{sec:4} and \ref{sec:5} on the asymptotic behavior of complex LS model and LHS model, respectively. In all simulations, we set
\begin{align*}
N = 50, \quad \Delta t = 0.02, \quad z_i \in \mathbb H\mathbb S^2, \quad i = 1,\cdots,50,
\end{align*}
and used the fourth order Runge-Kutta method.

\subsection{The complex Lohe sphere model}
In this subsection, we observe analytical results on the complex Lohe sphere model \eqref{D-1} in Section \ref{sec:4}.

\subsubsection{Identical case}
In this part, we observe the emergent behavior of \eqref{D-1} with all natural frequencies are equal. That is, there exists $\Omega \in \bbc^{3\times3}$ such that
\begin{align*}
\Omega^\dagger = -\Omega, \quad \Omega_i = \Omega, \quad i = 1, \cdots, 50.
\end{align*}
We perform a numerical simulation under plausible condition to observe the complete aggregation in Theorem \ref{T4.1}. In  all  simulations, we chose coupling strength, natural frequency $\Omega$ and frustration $W_0$ satisfying
\begin{align*}
\begin{aligned}
& \kappa_0 = 1, \quad \mbox{Re}(\Omega_{jk}), \mbox{Im}(\Omega_{jk}) \in [-1, 1], \\
& \mbox{Re}[(W_0)_{jk}], \mbox{Im}[(W_0)_{jk}] \in [-0.1, 0.1], \quad 1 \leq j,k \leq 3,
\end{aligned}
\end{align*}
and take initial data satisfying \eqref{C-1-2}. Under those settings, we observe the dynamics in time interval $[0, 10]$.

\begin{figure}[h]
\centering
\mbox{
\subfigure[~Graph of $\mathcal J_M$]{\includegraphics[scale = 0.19]{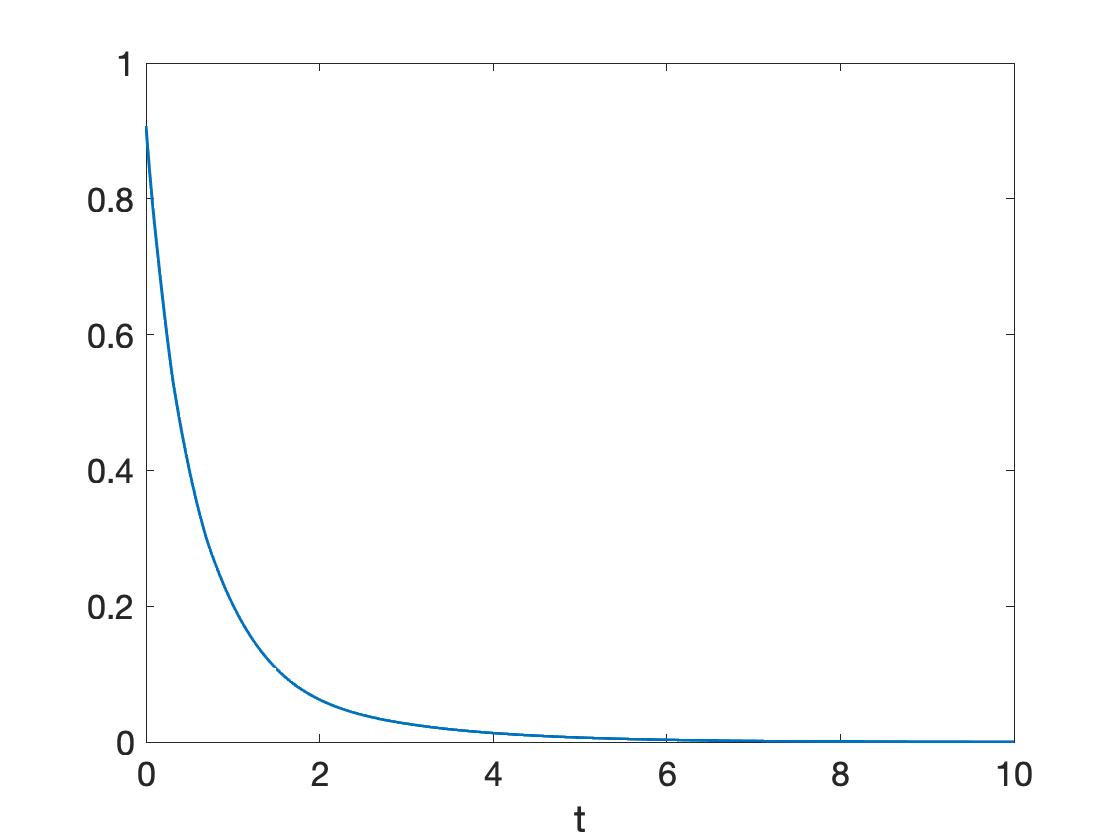}}
\subfigure[~Graph of $\log\mathcal J_M$]{\includegraphics[scale = 0.19]{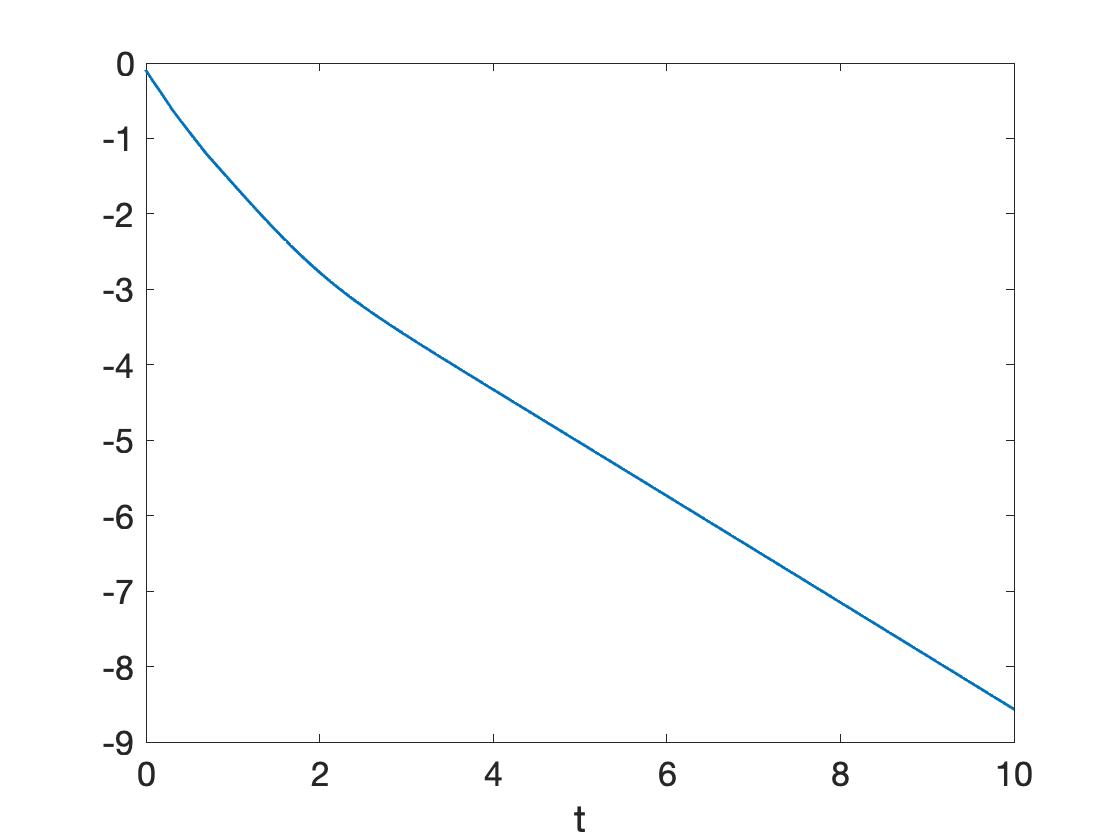}}
}
\caption{Emergence of complete aggregation}
\label{F1}
\end{figure}

In Figure \ref{F1} (a), we plot the graph of $\mathcal J_M$, from which one can observe the complete aggregation when it converges to zero. Furthermore, linearity of Figure \ref{F1} (b) exhibits the exponential decaying of $\mathcal J_M$.

\subsubsection{Nonidentical case}
In this part, we observe the emergent behavior of \eqref{D-1} with distinct natural frequencies. We perform numerical simulation under plausible condition to observe the practical aggregation in Theorem \ref{T4.2}. In simulation, we chose natural frequencies $\Omega_i$ and frustration $W_0$ satisfying
\begin{align*}
& \mbox{Re}[(\Omega_i)_{jk}], \mbox{Im}[(\Omega_i)_{jk}] \in [-1, 1], \quad \mbox{Re}[(W_0)_{jk}], \mbox{Im}[(W_0)_{jk}] \in [-0.1, 0.1], \quad 1 \leq j,k \leq 3,
\end{align*}
take initial data satisfying \eqref{C-1-2} and observe the dynamics in time interval $[0, 60]$. Under those settings, we employ various coupling strengths:
\begin{align*}
\kappa_0 = 1, ~5, ~10, ~20.
\end{align*}

\begin{figure}[h]
\centering
\mbox{
\subfigure[~Graphs of $\mathcal J_M^{\kappa_0}$ for various $\kappa_0$]{\includegraphics[scale = 0.19]{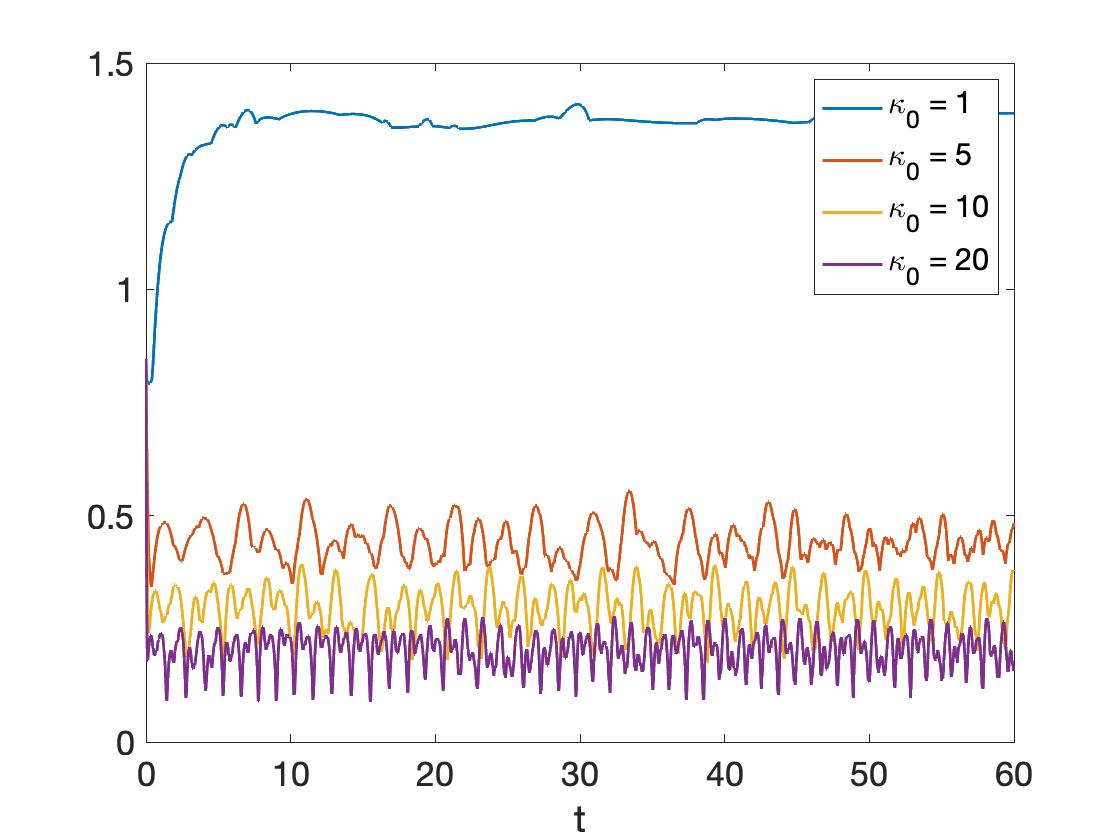}}
\subfigure[~Graphs of $\mathcal J_M^1$, $\sqrt{5}\mathcal J_M^5$ and $5\mathcal J_M^5$]{\includegraphics[scale = 0.19]{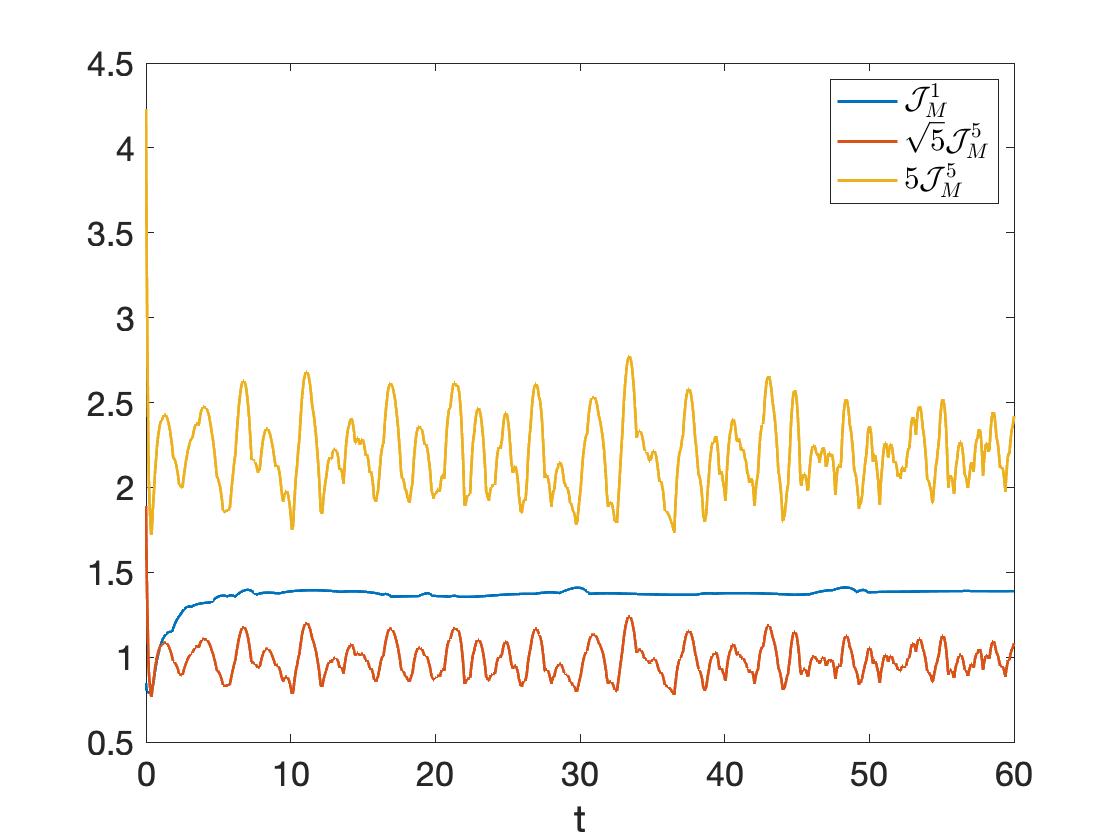}}
} \\ \mbox{
\subfigure[~Graphs of $\mathcal J_M^1$, $\sqrt{10}\mathcal J_M^{10}$ and $10\mathcal J_M^{10}$]{\includegraphics[scale = 0.19]{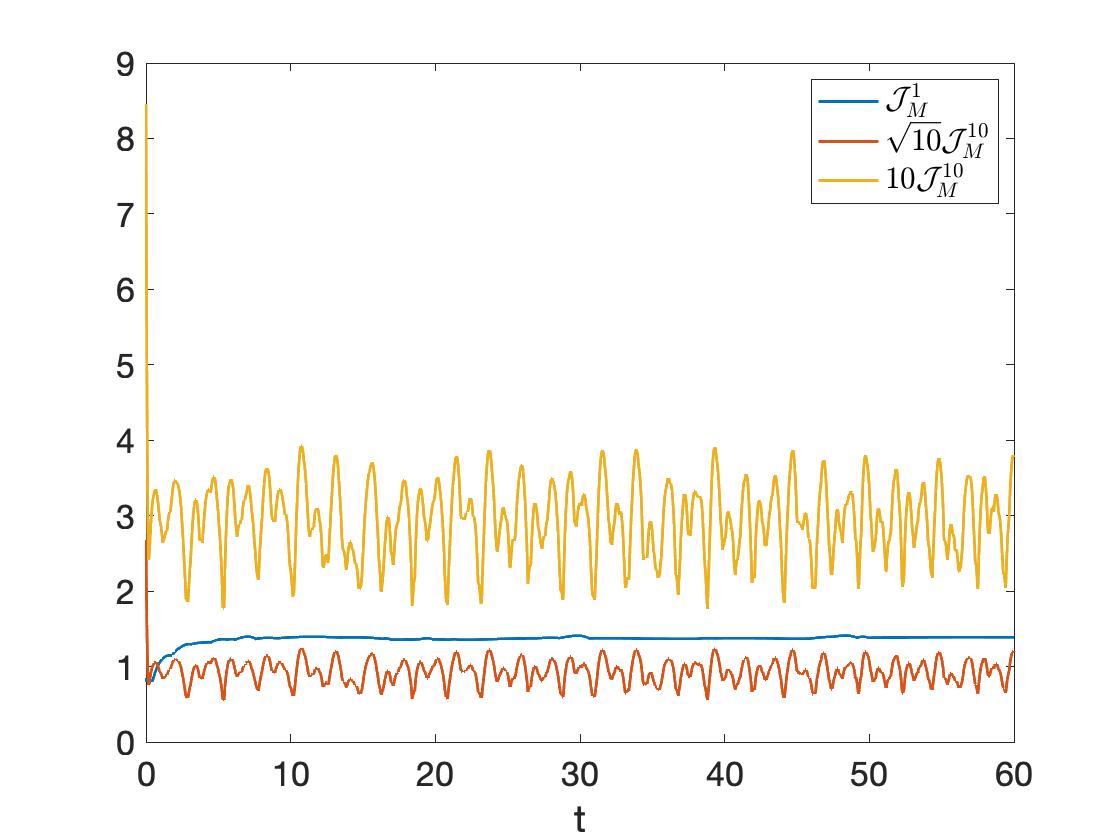}}
\subfigure[~Graphs of $\mathcal J_M^1$, $\sqrt{20}\mathcal J_M^{20}$ and $20\mathcal J_M^{20}$]{\includegraphics[scale = 0.19]{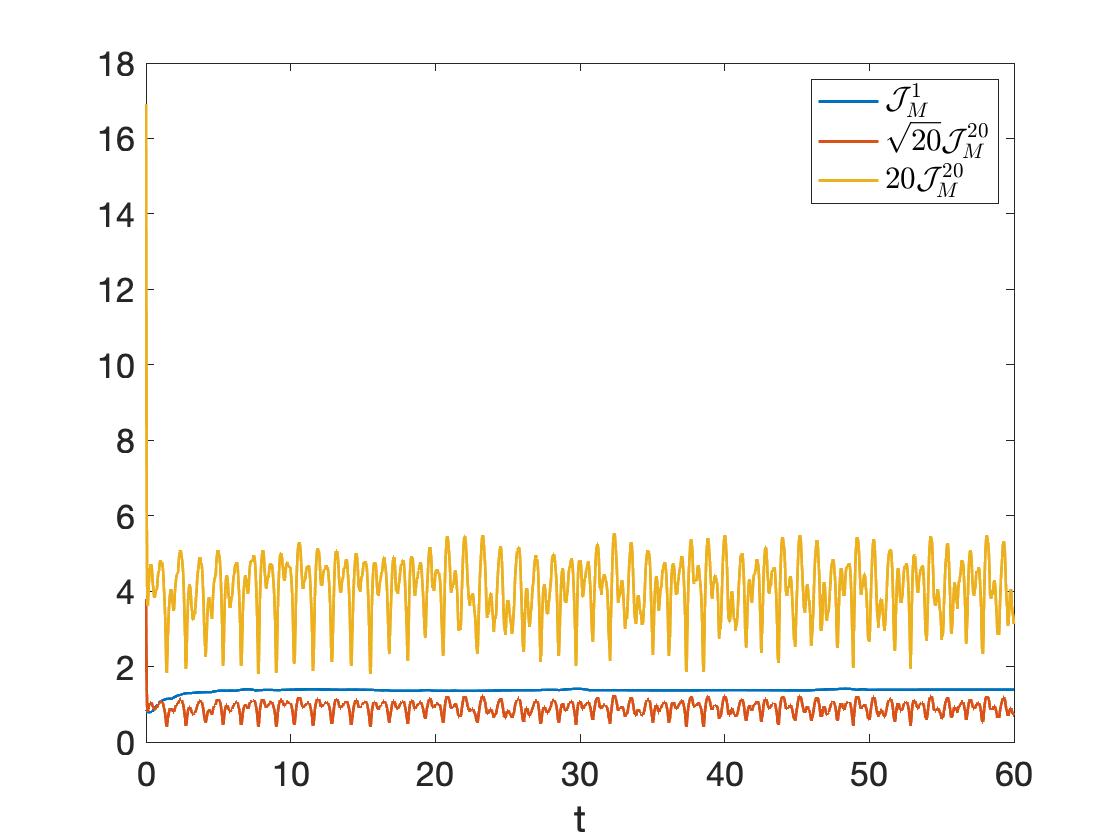}}
}
\caption{Emergence of practical aggregation}
\label{F3}
\end{figure}

In Figure \ref{F3} (a), we plot the graph of $\mathcal J_M$ for various $\kappa_0$. For each $\kappa_0$, we denote the corresponding $\mathcal J_M$ as $\mathcal J_M^{\kappa_0}$. Then, one can observe that the asymptotic bound of $\mathcal J_M^{\kappa_0}$ gets closer to zero as $\kappa_0$ becomes larger, which implies the practical aggregation. Moreover, in Figure \ref{F3} (b), (c) and (d), we plot the graphs of $\mathcal J_M^1$, $\sqrt{\kappa_0}\mathcal J_M^{\kappa_0}$ and $\kappa_0\mathcal J_M^{\kappa_0}$ for $\kappa_0 = 5, 10, 20$, respectively. One can observe that the graph of $\mathcal J_M^1$ is bounded by the graphs of $\sqrt{\kappa_0}\mathcal J_M^{\kappa_0}$ and $\kappa_0\mathcal J_M^{\kappa_0}$. More precisely, asymptotic bound of $\mathcal J_M^1$ is larger than that of $\sqrt{\kappa_0}\mathcal J_M^{\kappa_0}$, while less than that of $\kappa_0\mathcal J_M^{\kappa_0}$. And these observations support Remark \ref{R4.2}.

\subsection{The LHS model}
In this subsection, we observe the analytical result on the Lohe Hermitian sphere models in Section \ref{sec:5}.

\subsubsection{Identical ensemble}
In this part, we observe the emergent behavior of \eqref{LHSF} with all natural frequencies are equal. We perform numerical simulation under plausible condition to observe the complete aggregation in Theorem \ref{T5.1}. In simulation, we chose coupling strength, natural frequency $\Omega$ and frustrations $W_1, W_0$ satisfying
\begin{align*}
& \kappa_0 = 4, \quad \kappa_1 = 1, \quad W_1 = 0, \quad \mbox{Re}[(W_0)_{jk}],\ \mbox{Im}[(W_0)_{jk}] \in [-0.1, 0.1], \\
& \mbox{Re}(\Omega_{jk}), \mbox{Im}(\Omega_{jk}) \in [-1, 1], \quad 1 \leq j,k \leq 3,
\end{align*}
so that $\kappa_0 > 2\kappa_1$, and take initial data satisfying \eqref{C-1-4}$_4$. Under those settings, we observe the dynamics in time interval $[0, 5]$.

\begin{figure}[h]
\centering
\mbox{
\subfigure[~Graph of $\mathcal J_M$]{\includegraphics[scale = 0.19]{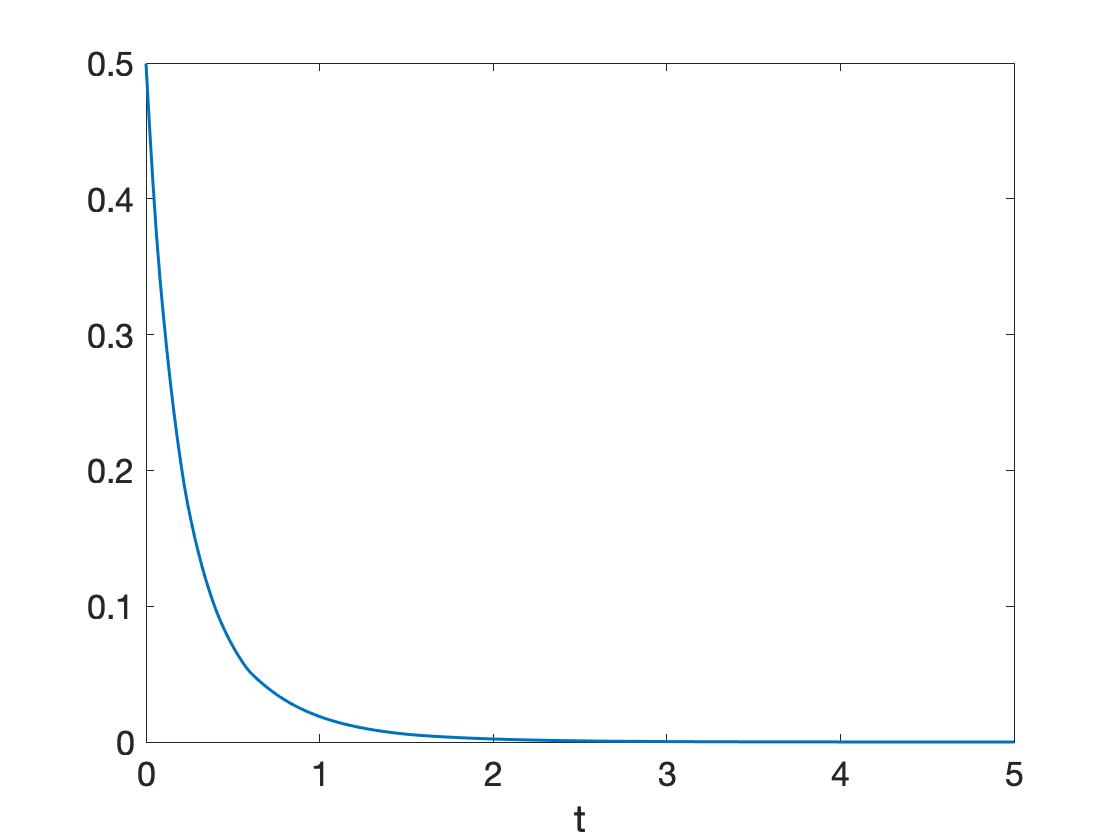}}
\subfigure[~Graph of $\log\mathcal J_M$]{\includegraphics[scale = 0.19]{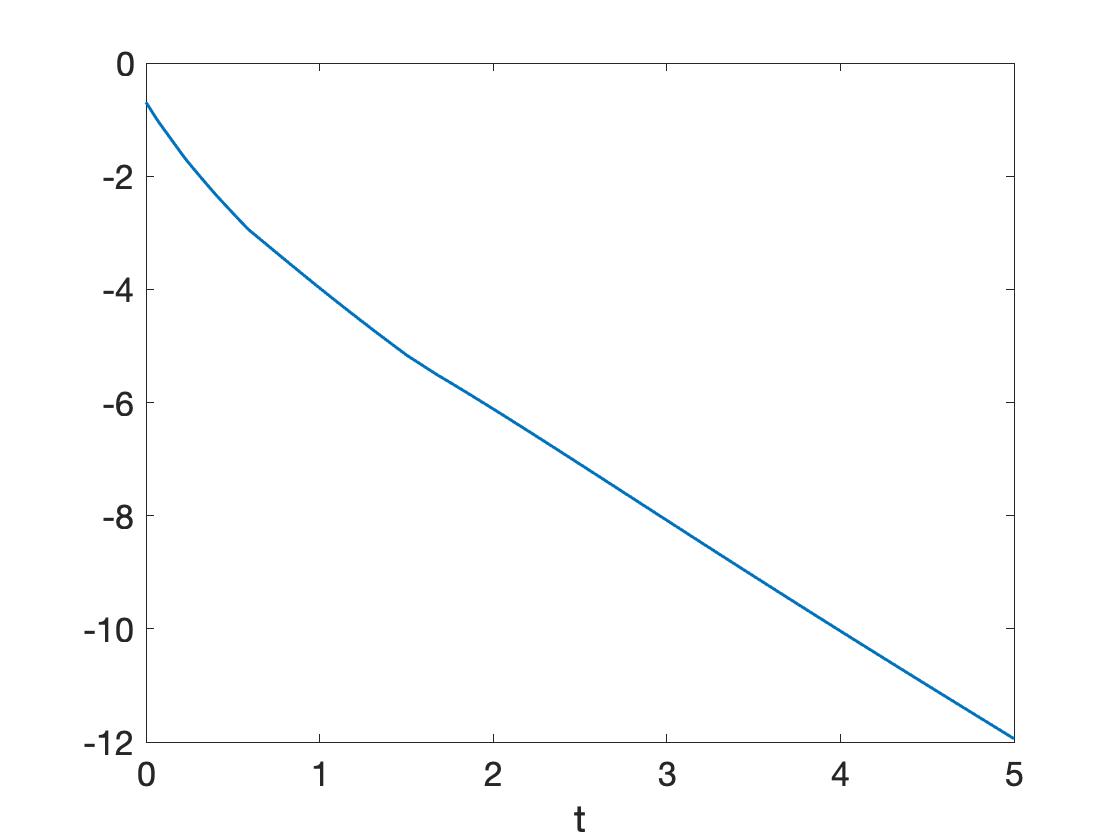}}
}
\caption{Emergence of complete aggregation}
\label{F2}
\end{figure}

In Figure \ref{F2} (a), we plot the graph of $\mathcal J_M$, from which one can observe the complete aggregation when it converges to zero. Furthermore, linearity of Figure \ref{F2} (b) exhibits the exponential decaying of $\mathcal J_M$.

\subsubsection{Nonidentical ensemble}
In this part, we observe the emergent behavior of \eqref{LHSF} with distinct natural frequencies. We perform numerical simulation under plausible condition to observe the practical aggregation in Theorem \ref{T5.2}. In simulation, we chose natural frequencies $\Omega_i$ and frustration $W_0, W_1$ satisfying
\begin{align*}
& \mbox{Re}[(\Omega_i)_{jk}], \mbox{Im}[(\Omega_i)_{jk}] \in [-1, 1], \quad \mbox{Re}[(W_l)_{jk}], \mbox{Im}[(W_l)_{jk}] \in [-0.1, 0.1], \\
& 1\leq i\leq N, \quad 1 \leq j,k \leq d, \quad l = 0, 1
& \end{align*}
take initial data satisfying \eqref{C-1-2} and observe the dynamics in time interval $[0, 100]$. Under those settings, we employ various coupling strengths:
\begin{align*}
(\kappa_0, \kappa_1) = (1, 1), ~(5, 1), ~(10, 1), ~(20, 1).
\end{align*}

\begin{figure}[h]
\centering
\mbox{
\subfigure[~Graphs of $\mathcal J_M^{\kappa_0}$ for various $\kappa_0$]{\includegraphics[scale = 0.19]{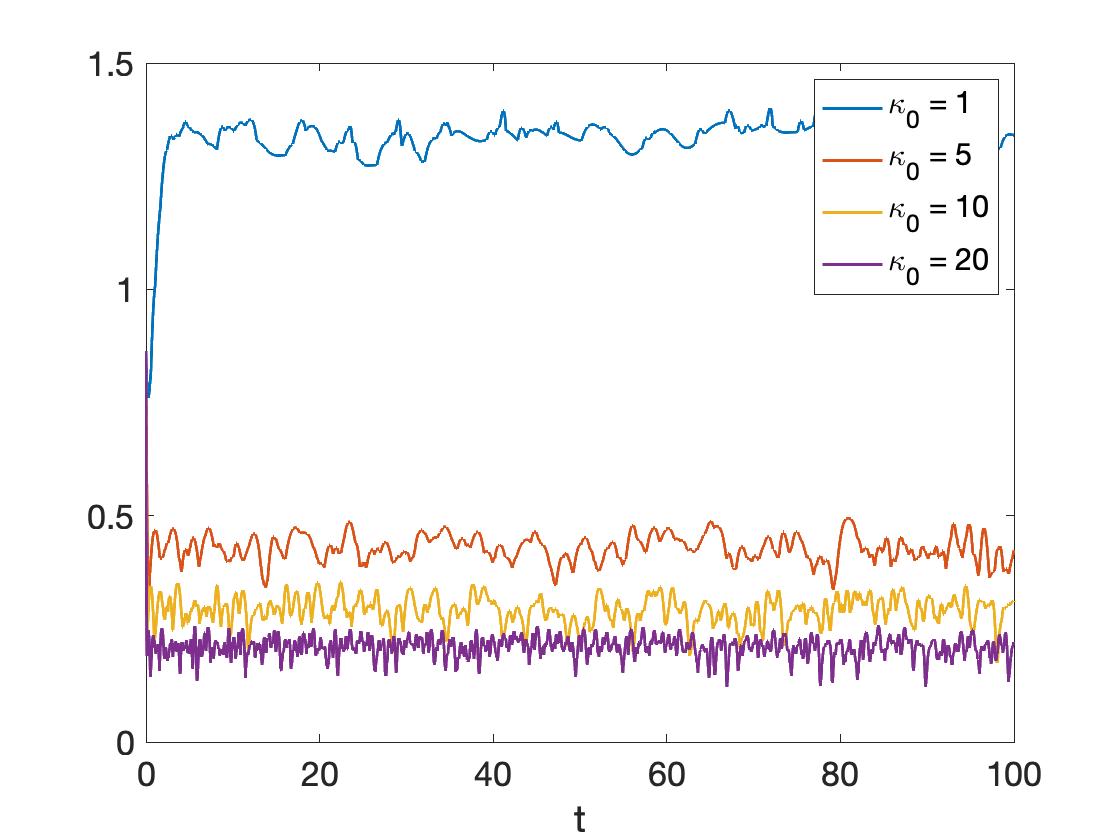}}
\subfigure[~Graphs of $\mathcal J_M^1$, $\sqrt{5}\mathcal J_M^5$ and $5\mathcal J_M^5$]{\includegraphics[scale = 0.19]{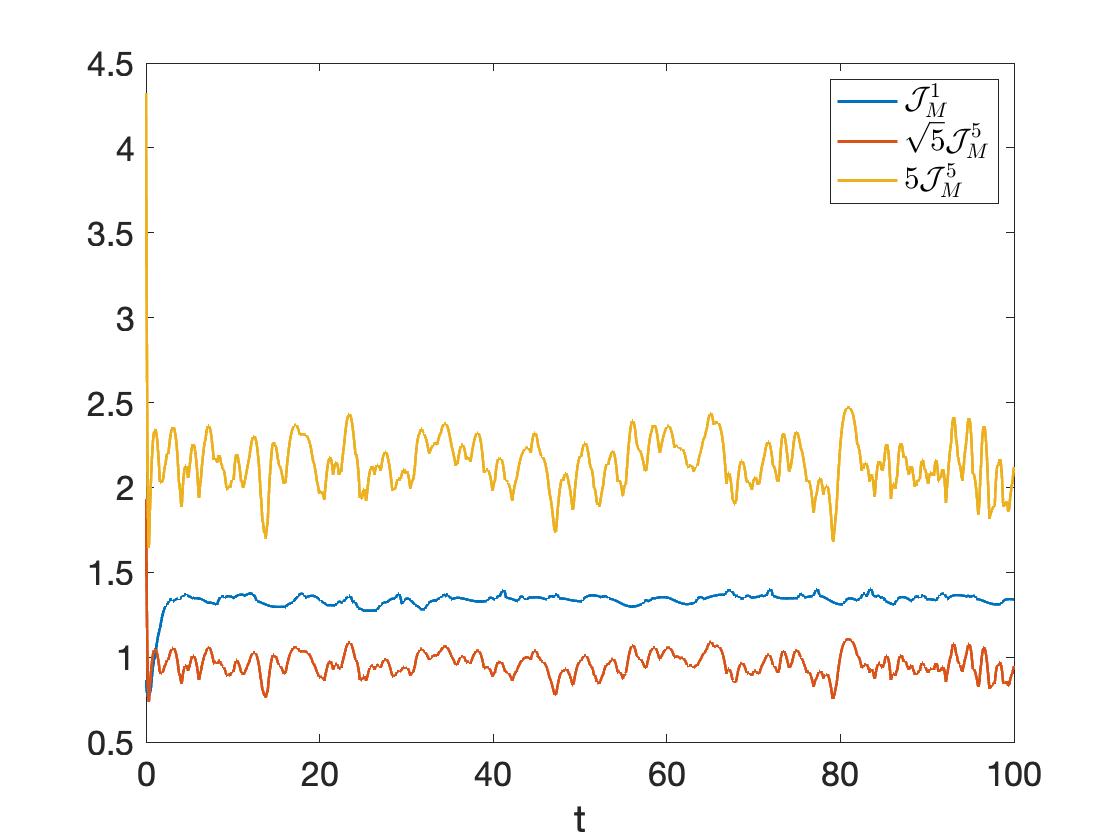}}
} \\ \mbox{
\subfigure[~Graphs of $\mathcal J_M^1$, $\sqrt{10}\mathcal J_M^{10}$ and $10\mathcal J_M^{10}$]{\includegraphics[scale = 0.19]{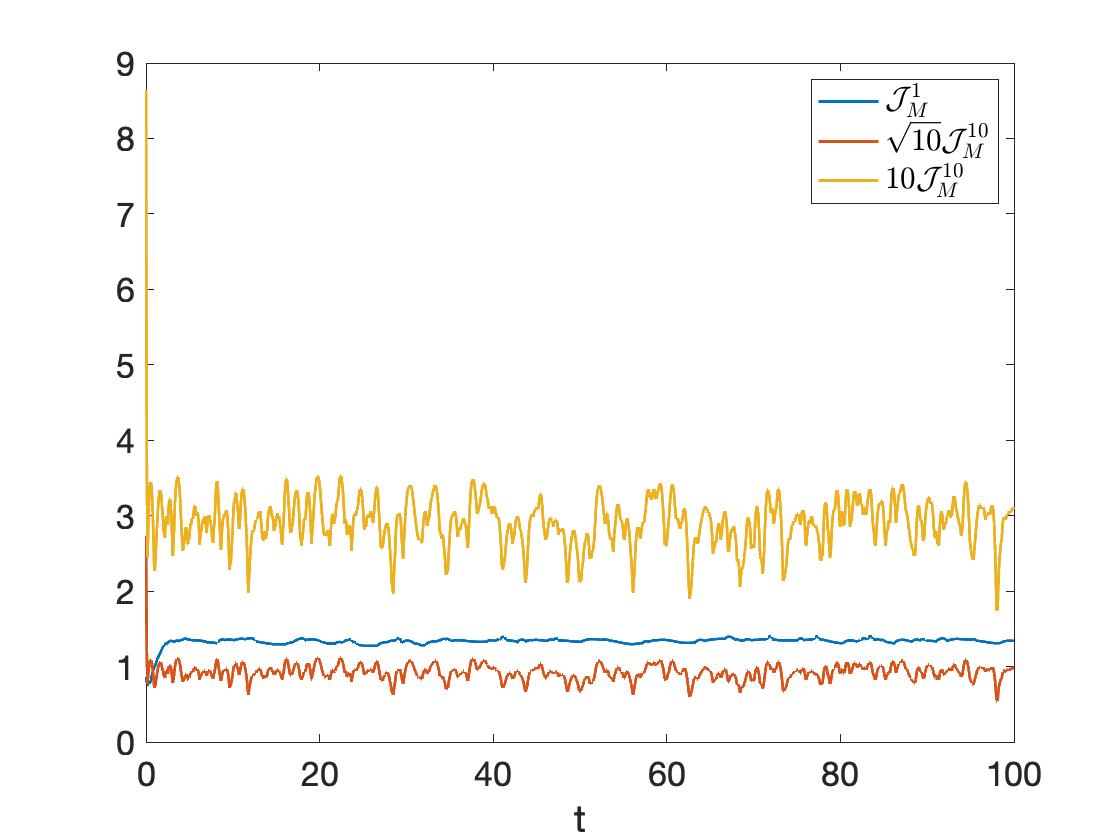}}
\subfigure[~Graphs of $\mathcal J_M^1$, $\sqrt{20}\mathcal J_M^{20}$ and $20\mathcal J_M^{20}$]{\includegraphics[scale = 0.19]{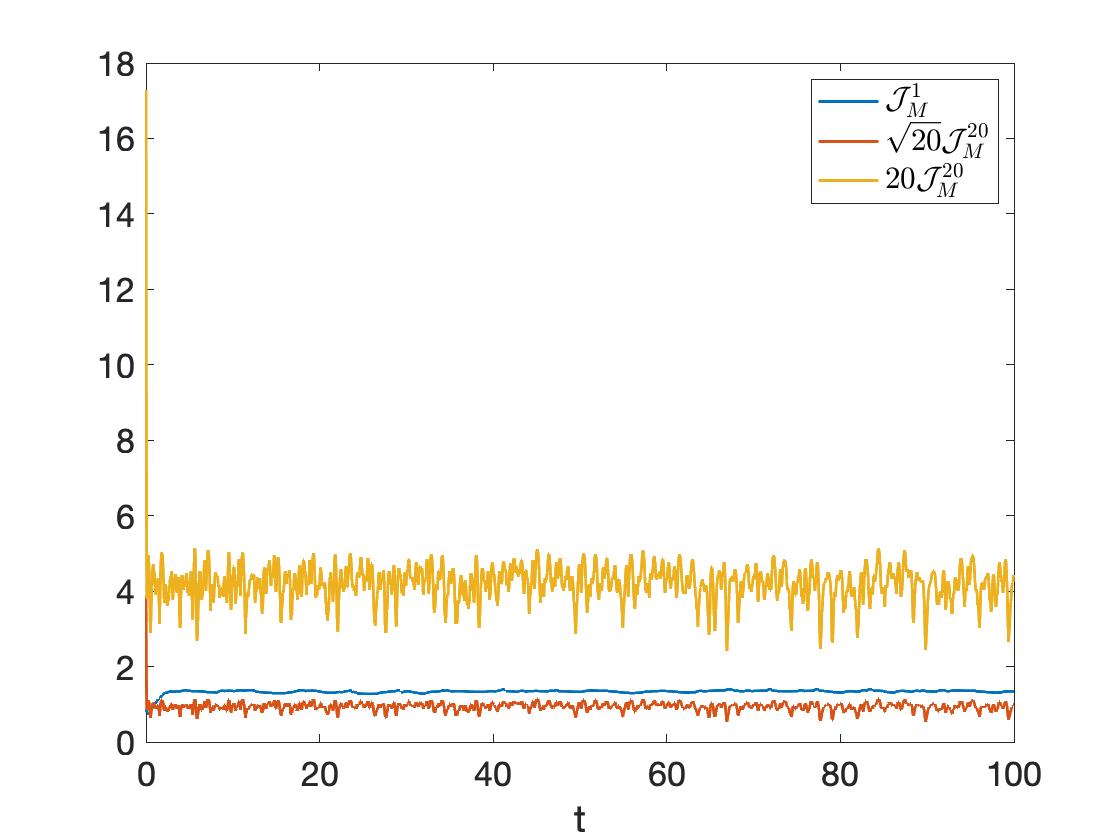}}
}
\caption{Emergence of practical synchronization}
\label{F4}
\end{figure}

In Figure \ref{F4} (a), we plot the graph of $\mathcal J_M$ for various $\kappa_0$, but $\kappa_1$ fixed. For each $\kappa_0$, we denote the corresponding $\mathcal J_M$ as $\mathcal J_M^{\kappa_0}$. Then, one can observe that the asymptotic bound of $\mathcal J_M^{\kappa_0}$ gets closer to zero as $\kappa_0$ becomes larger, which implies the practical synchronization. Moreover, in Figure \ref{F4} (b), (c) and (d), we plot the graphs of $\mathcal J_M^1$, $\sqrt{\kappa_0}\mathcal J_M^{\kappa_0}$ and $\kappa_0\mathcal J_M^{\kappa_0}$ for $\kappa_0 = 5, 10, 20$, respectively. One can observe that the graph of $\mathcal J_M^1$ is bounded by the graphs of $\sqrt{\kappa_0}\mathcal J_M^{\kappa_0}$ and $\kappa_0\mathcal J_M^{\kappa_0}$. More precisely, asymptotic bound of $\mathcal J_M^1$ is larger than that of $\sqrt{\kappa_0}\mathcal J_M^{\kappa_0}$, while less than that of $\kappa_0\mathcal J_M^{\kappa_0}$. And these observations support Remark \ref{R5.1}.

\section{conclusion} \label{sec:7}
\setcounter{equation}{0} 
In this paper, we have studied emergent behaviors of the Lohe hermitian sphere model with frustration. Frustration can act as an anti-aggregation mechanism so that it prevents a formation of aggregate phenomenon. However, when coupling is strong enough and frustration is small enough, the complete aggregation and practical aggregation can emerge depending on the nature of frequency matrices. For both cases, we provide explicit sufficient frameworks in terms of coupling strengths and initial data, and then within our proposed framework, we show that the LHS model can exhibit collective behaviors. For non-identical particles with different free flows, our estimates on the practical aggregation can be viewed as a weak aggregation estimate in the sense that our practical aggregation estimate does not give us on the formation of phase-locked states for the LHS model. Moreover, we do not know how many phase-locked states with positive order parameter can exist in a generic setting, and whether periodic orbits can exhibit in the LHS model is not known yet. For the identical ensemble, nontrivial periodic orbits will be excluded due to the constancy of cross-ratio-like quantity. Of course, aforementioned questions on the phase-locked states are also open, even for the Lohe sphere model. These interesting questions will be left for a future work. 

\end{document}